%% file: 1script.tex
\documentclass[review]{elsarticle}

\input{0packages}

\usepackage{hyperref,enumitem}
\usepackage[switch, modulo]{lineno}
\usepackage[scaled=.80]{beramono}

\journal{Journal of \LaTeX\ Templates}

\begin{document}
\input{0shortcuts}

\begin{frontmatter}

\title{Multi-Dimensional Unlimited Sampling\\ and Robust Reconstruction}
\tnotetext[mytitlenote]{D.~Florescu and A.~Bhandari are with the Department of Electrical and Electronic Engineering, Imperial College London, SW72AZ, UK.
This work was supported by the UK Research and Innovation council's Future Leaders Fellowship program ``Sensing Beyond Barriers'' (MRC Fellowship award no.~MR/S034897/1). E-mails: \{D.Florescu, A.Bhandari\}@imperial.ac.uk or ayush@alum.MIT.edu. Project page for (future) release of hardware design, code and data: \href{https://bit.ly/USF-Link}{\texttt{https://bit.ly/USF-Link}}.}

\author{Dorian Florescu and Ayush Bhandari}

\begin{abstract}
In this paper we introduce a new sampling and reconstruction approach for multi-dimensional analog signals. Building on top of the Unlimited Sensing Framework (USF), we present a new
folded sampling operator called the multi-dimensional modulo-hysteresis that is also backwards compatible with the
existing one-dimensional modulo operator. Unlike previous approaches, the proposed model
is specifically tailored to \MD signals. In particular, the model uses certain redundancy in dimensions $2$ and above, which is exploited for input recovery with robustness. We prove that the new operator is well-defined and its outputs have a bounded dynamic range. For the noiseless case, we derive a theoretically guaranteed input reconstruction approach. When the input is corrupted by Gaussian noise, we exploit redundancy in higher dimensions to provide a bound on the error probability and show this drops to $0$ for high enough sampling rates leading to new theoretical guarantees for the noisy case. Our numerical examples corroborate the theoretical results and show that the proposed approach can handle a significantly larger amount of noise compared to USF. 
\end{abstract}


\end{frontmatter}


\newpage

\tableofcontents

\newpage

\section{Introduction}

Shannon's sampling theory is the workhorse of almost all modern-world digital systems. Its practical implementation is carried out via electronic hardware, namely, the analog-to-digital converter (ADC). However, there is a gap between theory and practice which leads to a few fundamental deviations from the ideal sampling model, including, among others, quantization (see the extensive survey by Gray \& Neuhoff \cite{Gray:1998:J}) non-pointwise sampling \cite{Sun:2002:Ja} and ADC saturation. The latter deviation arises from the fact that the ADC is a physical device and hence, one can only record a fixed range of amplitudes (typically, a prescribed voltage range). This input amplitude range defines the \emph{dynamic range} (or DR) of the ADC, say $\lambda>0$. Any signal exceeding (in absolute value) $\lambda$ would result in permanent loss of information due to saturation or clipping. Mathematically, this is synonymous to  \emph{hard thresholding} \cite{Donoho:1994:J,Blumensath:2009:J}, but the difference is that it occurs in hardware and is highly undesirable. Clipped sample values lead to high frequency components, which in turn leads to aliasing. Typical solutions to the saturation problem rely on:
\begin{enumerate}[ label = (\alph*)]
\item hardware approaches such us \emph{companding} \cite{Smith:1957:J} or adaptively matching the dynamic range to the input signal range via \emph{automatic gain control}. There are also techniques that re-think ADC design (see, for example, \cite{Smaragdis:2009:J}). 
\item algorithmic approaches that aim to solve the inverse problem of \emph{de-clipping} \cite{Abel:1991,Zhang:2016} or \emph{inpainting} \cite{Adler:2012:J}. 
\end{enumerate}
Clipping or saturation is also highly relevant in the context of digital imaging, so much so that almost all modern smartphones are equipped with the \emph{High Dynamic Range} or ``HDR'' mode,  based on multiple captures that are combined numerically \cite{Debevec:1997:C}.

The progress in the last many decades has led to deepened understanding of the nuances involved with the quantization and limited DR aspects. Clearly, ADCs need to be matched to the DR of the input signal to avoid saturation or clipping. Beyond this calibration step---typically addressed by the engineers---there is an additional challenge: higher dynamic range requires a higher number of bits to achieve a given resolution; this in turn leads to a higher power consumption in the ADC, thus highlighting the integral role of DR in digital acquisition.

\begin{figure}[!t]
    \centering
    \includegraphics[width=1\textwidth]{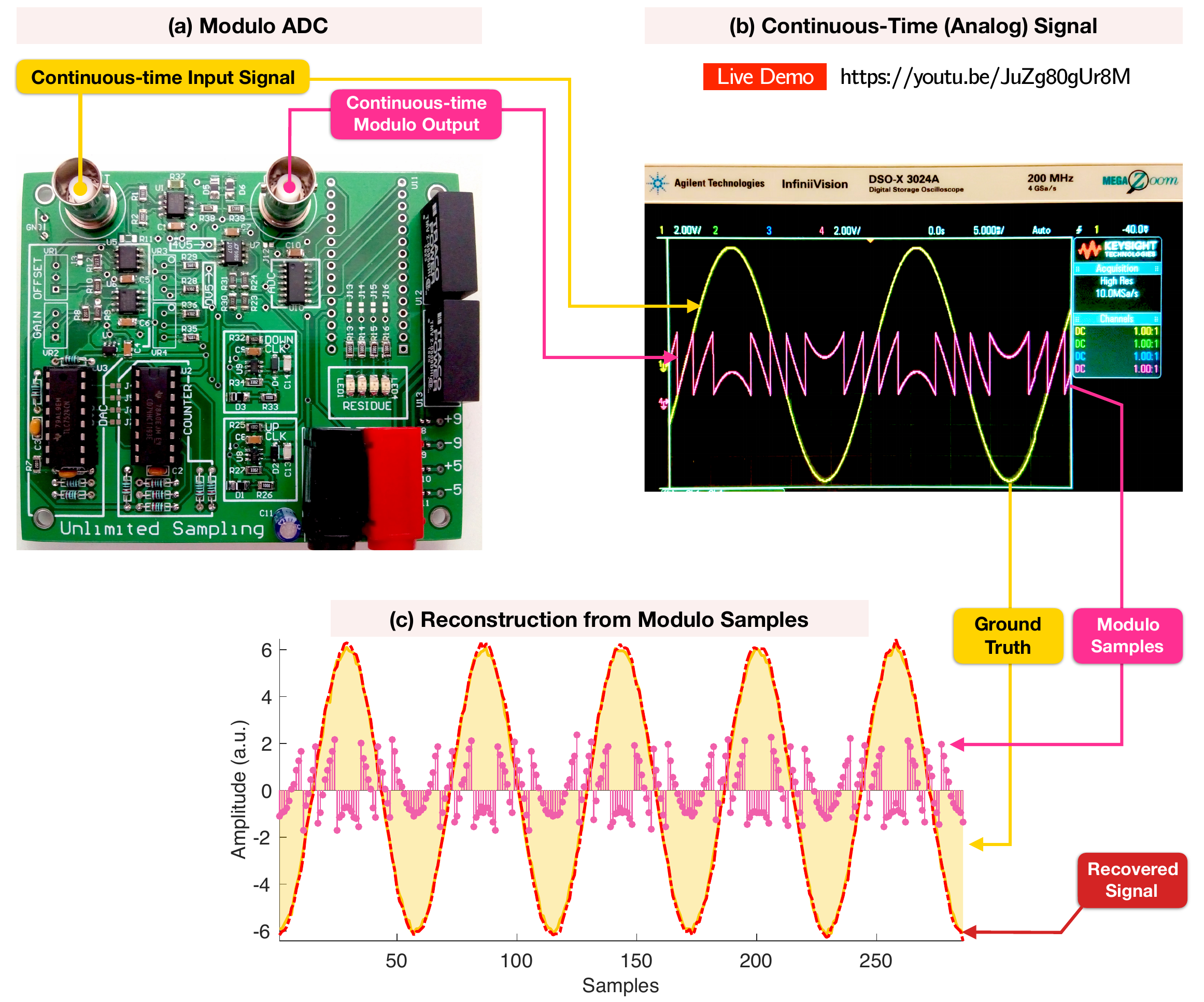}
    \caption{Signal acquisition and recovery pipeline for Unlimited Sensing Framework \cite{Bhandari:2017:C,Bhandari:2020:Ja}. (a) Modulo sampling hardware \cite{Bhandari:2021:J}. (b) Continuous-time signal waveforms on an oscilloscope. (c) Signal reconstruction using recovery algorithms \cite{Bhandari:2017:C,Bhandari:2020:Ja,Bhandari:2021:J,Florescu:2022:J}.}
\label{fig:USFPipeline}
\end{figure}

\subsection{Unlimited Sensing Framework (USF)}
Recently, the Unlimited Sensing Framework (USF) \cite{Bhandari:2017:C,Bhandari:2018:Ca,Bhandari:2018:C,Bhandari:2019:C,Bhandari:2020:Pata,Bhandari:2020:Ja,Bhandari:2021:J} has been proposed in the literature that serves as an alternative digital acquisition protocol for avoiding the DR limitation in conventional ADCs. The USF is based on a joint design of hardware and mathematical algorithms. 

\begin{itemize}
  \item In hardware, the modulo non-linearity ensures that HDR inputs are folded back in to the ADC's DR; this is because the modulo threshold is chosen such that the modulo ADC's range is bounded by $\lambda$. Consequently, the modulo ADC results in \emph{folded samples}. 
  \item To recover the HDR input from folded, modulo samples, mathematically guaranteed recovery algorithms are deployed. 
\end{itemize}

Similar to the Shannon--Nyquist sampling criterion where a higher input bandwidth can be traded off for higher sampling rates, it was shown that HDR signals can also be tackled by sampling more densely. This is made precise by the following theorem. 
\def\e{\mathrm{e}}
\begin{theo}[Unlimited Sampling Theorem \cite{Bhandari:2017:C}]
\label{th:USF}
Let $f(t)$ be a continuous-time function with maximum frequency $\Omega$ (rads/s). Then, a sufficient condition for recovery of $f( t )$ from its modulo samples (up to an additive constant) taken every $T$ seconds apart is $T\leqslant 1/ \left( 2 \Omega \e\right)$ where $\e$ is Euler's number.
\end{theo}

Thus, the USF addresses a major bottleneck in physical sensors by allowing the recovery of inputs beyond the sensor dynamic range. A first validation of the USF with experiments based on a modulo ADC were presented in \cite{Bhandari:2021:J}. In particular, it was shown that signals as large as $\approx 25\lambda$ can be recovered in a laboratory setup. The full sampling and reconstruction pipeline for the USF is shown in \fig{fig:USFPipeline}.

The initial works based on USF tackled signals supported on the real line spanned in a bandlimited \cite{Bhandari:2017:C,Bhandari:2020:Ja} or spline spaces \cite{Bhandari:2020:Ca}. There are also methods {to recover signals with compressive priors} \cite{Shah:2019:C,Musa:2018:C} and using wavelet filters \cite{Rudresh:2018:C}. Further extensions of the  USF include compactly supported inputs \cite{Bhandari:2021:J}, sparse signals \cite{Bhandari:2022:J} and also new acquisition models \cite{Florescu:2022:J}. 

{A new acquisition model called modulo-hysteresis was introduced in \cite{Florescu:2021:C} and further discussed in \cite{Florescu:2022:J}, which considers  hardware non-idealities and enables new recovery guarantees. The modulo-hysteresis was also implemented in a hardware prototype \cite{Florescu:2022:J}.}
This line of work also paved the path to novel and exciting low-power acquisition neuromorphic applications \cite{Florescu:2021:C,Florescu:2022:Ca}. 

The methods discussed so far assume that the input is \D\hspace{-0.2em}. However, in many applications, such as photography \cite{Bhandari:2020:Ca}, X-ray imaging or Computed Tomography \cite{Beckmann:2022:J}, the input signal is \MD.

\paragraph{Motivation for a \MD model} 
There have been attempts to address \MD inputs with modulo architectures by rasterizing and processing the signal line-by-line in imaging \cite{FernandezMenduina:2021:J,Bhandari:2020:Ca,Beckmann:2022:J} or for lattice sampling \cite{Bouis:2020:C}. However, the existing modulo sampling approaches for \MD data are based on a \D modulo operator that is applied sequentially on the slices of a \MD input. This considers each slice a distinct signal, and does not exploit that they are all part of a \MD input. In other words, this modulo operation represents a separable transformation that does not exploit the \MD nature of the input. Furthermore, it is known that non-separable transformations represent much more powerful tools in analysing \MD data \cite{Zhao:2018:J,Cohen:1993:J}.

We consider only the problem of input recovery for noisy inputs, distinct from that of input denoising, which was addressed before for modulo sampling \cite{Tyagi:2022:J,Fanuel:2021}. Methods such as USF recover the noise corrupted input samples while keeping the noise sequence intact. However, USF (Theorem \ref{th:USF}) is fundamentally restricted to work with noise amplitudes smaller than the modulo threshold. When this requirement is not satisfied, the input recovery is heavily distorted. This limitation is carried over to the existing attempts to apply USF to \MD data.

\paragraph{Contributions} 
Here we present a modulo model that exploits the \MD structure of the data in the encoding process. Specifically, via \MD sampling in $D$ dimensions, we are able to dedicate a $D-1$-dimensional subspace to deal with noise reduction, leaving dimension $d=1$ for estimating the modulo folds. Specifically, our contributions are below:
\begin{enumerate}[leftmargin = *,label = $\textrm{C}_{\arabic*})$]
\item We introduce a generalized $D$-dimensional modulo operator for sampling on a lattice.
\item We prove that the operator is well-defined and the folding discontinuities are located along directions given by the lattice vectors.
\item We provide recovery guarantees under noiseless assumption.
\item Under Gaussian noise assumption, we provide an upper bound on the recovery error probability that drops to $0$ for high enough sampling rates.
\item Using a numerical study we show that the proposed model offers significantly better noise robustness than USF.
\end{enumerate}

\input{0Notation.tex}

\newpage

\section{Modulo acquisition and recovery}

\subsection{Recovery from \D modulo data}
\label{subsect:USF}
The centered modulo with threshold $\lambda$ is a function $\MO:\mathbb{R}\rightarrow\mathbb{R}$ satisfying \cite{Bhandari:2020:Ja}
\begin{equation}
\label{eq:ideal_mod}
\MO \rb{x} = 2\lambda \left( {\fe{ {\frac{x}{{2\lambda }} + \frac{1}{2} } } - \frac{1}{2} } \right).
\end{equation}
When applied to a one-dimensional function $g$ the modulo non-linearity generates values $\MO \rb{g\rb{t}} \in \sqb{-\lambda,\lambda} $. 

In analogy to Shannon-Nyquist sampling theory, the first recovery result in the Unlimited Sensing Framework (USF) utilized bandlimited inputs, namely, $g \in \PW{\Omega}$. In the noiseless scenario, the \emph{unlimited sampling theorem} \cite{Bhandari:2017:C,Bhandari:2020:Ja} guarantees that the input of the ideal modulo encoder can be recovered from the output samples provided that the sampling period satisfies $T<\frac{1}{2\Omega e}$. Furthermore, reconstruction is also possible in the case of data corrupted by bounded noise if the following is true \cite{Bhandari:2020:Ja}
\begin{equation}
\label{eq:gamma_bound_hyst}
\rb{T\Omega e}^N g_\infty + 2^N \eta_\infty<\lambda.
\end{equation}

The recovery approach used in \cite{Bhandari:2017:C,Bhandari:2020:Ja} aims to reconstruct the residual function $\varepsilon_g \rb{t}$ defined as $\varepsilon_g \rb{t}\triangleq g\rb{t}-\MO \rb{g\rb{t}}\in 2\lambda\cdot \mathbb{Z}$. In other words, for the ideal modulo encoder the values of $\varepsilon_g \rb{t}$ lie on an equally spaced grid with step $2\lambda$. However, this is not true for non-ideal modulo encoders exhibiting phenomena such as hysteresis, leading to reconstruction distortions. 

A generalized model of the modulo operator, called \emph{modulo-hysteresis}, was introduced for the one-dimensional scenario \cite{Florescu:2022:C,Florescu:2022:Cb,Florescu:2022:J}. 
Here, we generalize this model to \MD sampling. As in the \D case, we will show that modulo-hysteresis enables the separation of the folding times, which will be used in the recovery in Section \ref{sect:MD}, \ref{sect:det_disc}, and \ref{sect:inp_rec}. We begin with the definition of the one-dimensional modulo-hysteresis.
\begin{definition}[One-dimensional modulo-hysteresis]
	The operator $\MOh$ with threshold $\lambda$ and hysteresis $h\in\left[0,2\lambda/3\right)$, where $\mathbf{\mathsf{H}}=\sqb{\lambda\ h}$, generates a function $z \rb{t}=\MOh g\rb{t}$ for input $g \in \pw$, such that, for $t\geq 0$
	\begin{equation}
		\label{eq:zt}
	    z\rb{t}=g\rb{t}-\varepsilon_g\rb{t},
	\end{equation}		
	where 
	\begin{itemize}
	    \item $\varepsilon_g\rb{t}=2\lambda_h\sum\nolimits_{r=1}^{R}s_r \ind_{[\tau_r,\infty)}\rb{t}+h M,\quad$ $\lambda_h\triangleq \lambda-h/2$,
	    \item $M=\floor{\tfrac{g\rb{0}+\lambda}{h}}-1$,
	    \item $\tau_r$ and $s_r$ are the folding time and sign respectively, satisfying $\tau_0=s_0=0$ and
	     
	\end{itemize}
	\begin{align}
	    \tau_1&=\inf\cb{t > \tau_0 \setsep \MO\rb{g\rb{t}+\lambda}= 0},\nonumber\\
	    s_r&=\mathrm{sign} \rb{g\rb{\tau_r}-g\rb{\tau_{r-1}}},
		\label{eq:folds_signs}\\
		\tau_{r+1}&=\inf\cb{t> \tau_r \setsep \MO\rb{g\rb{t}-g\rb{\tau_r}+h s_r}=0},\quad r\geq1.\nonumber
	\end{align}	
	Furthermore, for $t<0$ we have $z\rb{t}=\MOh\sqb{g\rb{-\cdot}}\rb{-t}$. Let $\tau^-_r,s^-_r,r\geq1$ be the sequence of folding times and signs computed via \eqref{eq:folds_signs} for $\MOh\sqb{g\rb{-\cdot}}$. Then we define $\tau_r\triangleq -\tau^-_{-r},s_r\triangleq s^-_{-r}, r\in\Z,r<0$.
	\label{def:cont_modulo}
\end{definition}

A key property of the 1D modulo-hysteresis is the folding time separation \cite{Florescu:2022:J,Florescu:2022:C}
\begin{equation}
\tau_{r+1}-\tau_r \geq \frac{h}{\Omega \norm{g}_\infty}.
\label{eq:separation}
\end{equation}
We note that the ideal modulo, which satisfies $\MO \rb{g\rb{t}}=\MOh g\rb{t}$ for $h=0$ does not guarantee any separation via \eqref{eq:separation}. The reconstruction problem proposed aims to recover $g\rb{kT}$ from $y\sqb{k}=z\rb{kT}$. Furthermore, it was shown that this approach enables handling a number of modulo non-idealities \cite{Florescu:2022:J,Florescu:2021:Ca,Florescu:2022:C,Florescu:2021:C}.

\section{Multi-dimensional modulo sampling}
\label{sect:MD}
\subsection{Multi-dimensional lattice sampling preliminaries}

Let $f:\mathbb{R}^D\rightarrow \mathbb{R}$ be a $D$-dimensional scalar function. The data is then sampled on a lattice $\bs{\Lambda}=\{\mathbf{V} \mathbf{T} \mathbf{k}\setsep \mathbf{k}\in\Z^D\}$.  We denote the resulting samples by $\gamma\sqb{\mathbf{k}}=f\rb{\mathbf{VTk}}$.
We assume that $f$ has a Fourier transform satisfying
\[
\mathrm{supp}\rb{\mathcal{F}f}\subseteq \mathbb{D}.
\]
Upon sampling, the spectrum of $f$ is copied periodically, to produce the \MD discrete-time Fourier transform $F_{\bs{\Lambda}}\rb{\bs{\omega}}=\sum_{\mathbf{k}\in\ZD}f\rb{\mathbf{VTk}}e^{-\jmath \inner{\bs{\omega},\mathbf{k}}}$, whose support satisfies
\begin{equation}
    \mathrm{supp}\rb{F_{\bs{\Lambda}}}=\mathbf{T}\mathbf{V}^\top\cdot \underset{\mathbf{n}\in\hat{\bs{\Lambda}}}{\bigcup} \rb{\mathbb{D}+2\pi\mathbf{n}}.
\end{equation}
It was shown that $f$ can be recovered from its lattice samples $\gamma\sqb{\mathbf{k}}$ if \cite{Viscito:1991:J,Lu:2009:J}
\begin{equation}
\label{eq:lattice_rec}
    \mathbf{T}\mathbf{V}^\top\rb{\mathbb{D}+2\pi\mathbf{n}_1}\cap\mathbf{T}\mathbf{V}^\top\rb{\mathbb{D}+2\pi\mathbf{n}_2}=\emptyset,\quad\forall \mathbf{n_1},\mathbf{n_2}\in\hat{\bs{\Lambda}}, \mathbf{n_1}\neq\mathbf{n_2}.
\end{equation}

To ensure that recovery is possible, we assume that the spectrum of $f$ has a compact support satisfying $\mathbb{D}\subseteq \hat{\mathbf{V}}\prod_{d=1}^D \rb{-\Omega_d,\Omega_d}$. Formally, our assumption is $f\in \pwd$.
Then $f$ can be reconstructed from samples $\gamma\sqb{\mathbf{k}}$ if \eqref{eq:lattice_rec} is satisfied, which is sufficiently guaranteed if we replace $\mathbb{D}$ by $\hat{\mathbf{V}}\cdot\prod_{d=1}^D \rb{-\Omega_d,\Omega_d}$, for which the terms on the left-hand-side of \eqref{eq:lattice_rec} are computed as
\begin{equation}
    \mathbf{T}\mathbf{V}^\top\rb{\mathbb{D}+2\pi\mathbf{n}}=\mathbf{T}\mathbf{V}^\top\sqb{\mathbf{V}^{-\top} \prod_{d=1}^D \rb{-\Omega_d,\Omega_d}+2\pi\mathbf{n}},\quad\mathbf{n}\in\hat{\bs{\Lambda}}.
\end{equation}
We use that $\mathbf{n}=\mathbf{V}^{-\top}\mathbf{T}^{-1}\mathbf{k},\quad \mathbf{k}\in\ZD$, yielding
\begin{align}
    \mathbf{T}\mathbf{V}^\top\rb{\mathbb{D}+2\pi\mathbf{n}}= \mathbf{T} &\prod_{d=1}^D \rb{-\Omega_d,\Omega_d}+2\pi\mathbf{k},\quad\mathbf{k}\in\ZD\\
    = &\prod_{d=1}^D \rb{-T_d\Omega_d,T_d\Omega_d}+2\pi\mathbf{k}.
\end{align}
Therefore, \eqref{eq:lattice_rec} is true if
\begin{equation}
    T_d<\frac{\pi}{\Omega_d}, \quad \forall d\in\cb{1,\dots,D}.
\end{equation}

Just as in the \D case, considering the problem of sensor saturation motivates using the concept of modulo folding also in the \MD case. Next we go through some of the attempts to apply modulo for \MD data.

\subsection{Previous approaches for recovery from \MD data}

As in the \D case, the problem with computing directly $\gamma\sqb{\mathbf{k}}$ is that the sample values may be very large which would saturate an analog-to-digital (ADC) acquisition device, which has a restricted dynamic range \cite{Bhandari:2020:Ja}. 
The modulo operator was applied previously for \MD inputs by processing a 1D slice at a time \cite{Bhandari:2020:Ca,Bouis:2020:C}. However, these methods are not truly \MD because they don't exploit the \MD structure of the data. Furthermore, when dealing with noise in 1D, the modulo samples $y\sqb{\mathbf{k}}=\gamma\sqb{\mathbf{k}}+\eta\sqb{\mathbf{k}}-\varepsilon_\gamma\sqb{\mathbf{k}}$ require the separation of both residual $\varepsilon_\gamma$ and $\eta$ within the same dimension. This turns out to be contradictory, as detecting $\varepsilon_\gamma$ requires a high-pass filter (such as $\Delta^N$ in the case of USF), while denoising is typically done with low-pass filters \cite{Florescu:2022:Cb}. 
Furthermore, a denoising approach on modulo data was tested for \MD signals \cite{Tyagi:2022:J}. However, we center our analysis on purely modulo inversion techniques, where the noise sequence remains unaltered. 

We define the following functions, representing slices of function $f(\mathbf{V}\mathbf{x})$. Let $f_{\mathbf{V}}\rb{\mathbf{x}} \triangleq f(\mathbf{V}\mathbf{x}) $. Let $f_{\bar{\mathbf{x}}}:\mathbb{R}\rightarrow \mathbb{R}$ denote the slice along dimension $x_1$ defined as $f_{\bar{\mathbf{x}}}\rb{x}\triangleq f\rb{x \mathbf{v}_1 + \sum_{d=2}^D x_d \mathbf{v}_d}, \forall \bar{\mathbf{x}}\in\mathbb{R}^{D-1},\bar{\mathbf{x}}=\sqb{x_2,\dots,x_D}$. In the next proposition we also use a generic slices defined as the \D function $g_d\rb{x_d}\triangleq f\rb{\sum_{n=1}^D x_n \mathbf{v}_n}$ by fixing dimensions $\cb{x_1,\dots,x_{d-1},x_{d+1},\dots,x_D}, d\neq 1$. The following proposition was proven in \cite{Bouis:2020:C}. 

\begin{prop}[Bandlimited slices]
The function $f_{\bbf{x}}$ satisfies $f_{\bbf{x}}\in \mathsf{PW}_{\Omega_1}\rb{\R}$. Furthermore, $g_d\in \mathsf{PW}_{\Omega_d}\rb{\R}$, $\forall d \in\cb{1,\dots,D}, d\neq 1$.
\label{prop:pw_for_slices}
\end{prop}
\if\BdlSlicesInAppendix\FL
    \begin{proof}
    \input{Proofs/BdlSlices}
    \end{proof}    
\else
    \begin{proof}
    The proof is in Section \ref{sect:proofs_prop}.
    \end{proof}
\fi

The proposition above proves that any \emph{slice} of the \MD function $f$ along lattice dimension $d$ has the spectrum compactly supported within $\rb{-\Omega_d,\Omega_d}$. Furthermore, this implies that a Bern\v{s}te\u{\i}n bound can be applied for each variable such that
\begin{align}
\begin{split}
    \vb{\frac{\partial}{\partial x_d}f_{\mathbf{V}}\rb{\mathbf{x}}}&\leq \Omega_d \max_{x_d}  \vb{f_{\mathbf{V}}\rb{\mathbf{x}}}\leq\Omega_d \norm{f_{\mathbf{V}}}_\infty\\
    &=\Omega_d \norm{f}_\infty,\quad \forall d\in\cb{1,\dots,D}, \forall \bbf{x}\in\RDred.
\end{split}
\label{eq:Bernstein_multiD}
\end{align}
The works in \cite{Bhandari:2020:Ca} and \cite{Bouis:2020:C} apply \D ideal modulo to $f_{\bar{\mathbf{x}}}$:
\begin{equation*}
\MO f_{\bar{\mathbf{x}}}\rb{x}=f_{\bar{\mathbf{x}}}\rb{x}-2\lambda_h \sum_{r\in\mathbb{Z}} s_{\bar{\mathbf{x}},r}\ind_{\left[\tau_{\bar{\mathbf{x}},r},\infty\right)}\rb{x}, x\in\mathbb{R}.
\end{equation*}
Therefore we can define the "folded" \MD function $z\rb{\mathbf{x}}$ as
\begin{equation}
    z\rb{\mathbf{Vx}}=\MO \rb{f_{\bbf{x}}\rb{x_1}}, 
\end{equation}
where $\mathbf{x}=\sqb{x_1,\dots,x_D}, \bbf{x}=\sqb{x_2,\dots,x_D}, x_d\in\R,\forall d\in\cb{1,\dots,D}$. Subsequently, the output samples are $z\rb{\mathbf{VTk}}=\MO f_{\bbf{T}\bbf{k}}\rb{k_1 T_1}$, where $\bbf{T}=\mathrm{diag}\cb{T_2,\dots,T_D}$.

Then, recovering $f_{\bbf{T}\bbf{k}}\rb{k_1 T_1}$ for all $\bbf{T}\bbf{k}$ from $z\rb{\mathbf{VTk}}$ represents a line-by-line approach used in \cite{Bhandari:2020:Ca,Beckmann:2022:J,Bouis:2020:C}, which is guaranteed to work if \eqref{eq:gamma_bound_hyst} holds true. However, as explained previously, this approach does not exploit the \MD structure of the input data. This is further motivated by the accepted knowledge in image processing that non-separability in multiple dimensions has a lot more to offer than separability \cite{Zhao:2018:J,Cohen:1993:J}.
This motivates introducing a \MD modulo-hysteresis model in the next section. We will show that the new model allows a significantly large amount of noise, which is not possible with USF that processes the data line-by-line.

\subsection{Towards multi-dimensional modulo-hysteresis acquisition}

Inspired from the 1D modulo-hysteresis operator that showed improvements for noise robustness \cite{Florescu:2022:Cb}, we define in the following a new operator called \MD modulo-hysteresis that addresses the issues discussed in the previous subsection. The idea is to split the domain $\RDred$ in disjoint sets confined in polytopes $\cP_{\bbf{b}}, \bbf{b}\in\ZDred$ defined as
\begin{equation}
    \cP_{\bbf{b}}=\cb{ \bbf{V} B \rb{\bbf{b}+\bbs{\alpha}} \setsep \forall \bbs{\alpha}\in\left[0,1\right)^{D-1}},
    \label{eq:multiD_intervals}    
\end{equation}
where $\bbf{b}=\sqb{b_2,\dots,b_D}$ and $B\in\R_+^*$ is the polytope edge length along directions parallel with versors $\mathbf{v}_2,\dots,\mathbf{v}_D$. The set $\cP_{\bbf{b}}$ is created via the last $D-1$ vectors of the lattice basis $\mathbf{v}_2,\dots,\mathbf{v}_D$ where the basis coordinates lie in a set of rectangular polytopes $\cR_{\bbf{b}}$ such that $\cP_{\bbf{b}}=\cb{\bbf{V}\bbf{x}\setsep \bbf{x}\in\cR_{\bbf{b}}}$,
where
\begin{equation}
    \cR_{\bbf{b}}=\cb{B\rb{\bbf{b}+\bbs{\alpha}}\setsep \bbs{\alpha}\in\left[0,1\right)^{D-1}}=\prod_{d=2}^D \left[ b_d B, \rb{b_d+1} B \right).
    \label{eq:multiD_rectangular_intervals}
\end{equation}
Note that $\RDred=\bigcup_{\bbf{b}\in\ZDred} \cP_{\bbf{b}}$. Sets $\cP_{\bbf{b}}$ thus can be used to split the domain of function $f$ in disjoint \emph{bands} of width $B$ given by $\cB_{\bbf{b}}=\rb{\R\mathbf{v}_1}\times\cP_{\bbf{b}}$ identified using the indices in $\bbf{b}$, such that $\RD=\bigcup_{\bbf{b}\in\ZDred} \cB_{\bbf{b}}$.
By exploiting the smoothness of $f$ we can derive that, for a fixed $x_1\in\R$, $f$ has bounded variation within each band, and thus the folding can occur simultaneously on all coordinates $x_2,\dots,x_D$, which gives the folded signal a particular structure to be exploited in recovery. The definition of the new operator is given as follows.

\begin{definition}[Multi-dimensional modulo-hysteresis]
The operator $\MOh^D$ with threshold $\lambda$ and hysteresis $h\in\left[0,2\lambda/3\right)$, where $\bs{\mathsf{H}}=\sqb{\lambda\ h}$, generates a function $z$  for input $f \in \pwd$ such that, for $\sqb{\mathbf{x}}_1 =x_1\geq 0$
\begin{equation}
	z\rb{\mathbf{Vx}}=f\rb{\mathbf{Vx}}-\varepsilon_f\rb{\mathbf{Vx}},
	\label{eq:zt_MD}
\end{equation}		
where $\varepsilon_f$ denotes the modulo-hysteresis residual defined as
\begin{equation}
    \varepsilon_f\rb{\mathbf{Vx}}=h\sqb{M_{\bbf{b}}+\sum_{r=0}^{R_{\bbf{b}}^+}s_{\bbf{b},r}\ind_{\left[\tau_{\bbf{b},r},\infty\right)}\rb{x_1}}, \forall\bbf{x}=\sqb{x_2,\dots,x_D}\in \cR_{\bbf{b}},
    \label{eq:eps}
\end{equation}
where $\cR_{\bbf{b}}$ satisfies \eqref{eq:multiD_rectangular_intervals}, $\mathbf{x}=\sqb{x_1,\dots,x_D}$, $\bbf{x}=\sqb{x_2,\dots,x_D}$, and $M_{\bbf{b}}\in\Z$ satisfies \begin{equation}
    M_{\bbf{b}}=\floor{\frac{\inf_{\bbf{x}\in \cR_{\bbf{b}}} f_{\bbf{x}}\rb{0}+\lambda}{h}}-1,
\end{equation}	 
and $\tau_{\bbf{b},r}\in \cb{0,\dots,R_{\bbf{b}}^+}, s_{\bbf{b},r}\in \cb{-1,0,1}$ are the folding times and signs in band $\mathcal{B}_{\bbf{b}}$, respectively, defined as
\begin{gather}
\tau_{\bbf{b},r+1}=\inf\cb{x_1>\tau_{\bbf{b},r}\setsep \sup_{\bbf{x}\in \cR_{\bbf{b}}}\vb{f_{\bbf{x}}\rb{x_1}-\varepsilon_{\bbf{b},r}\rb{x_1}}=\lambda},
\label{eq:tau_recurrent}\\
s_{\bbf{b},r+1}= \mathrm{sign} \sqb{f_{B\bbf{b}}\rb{\tau_{\bbf{b},r+1}}-\varepsilon_{\bbf{b},r}\rb{\tau_{\bbf{b},r+1}}},
\label{eq:sign_recurrent}\\
\varepsilon_{\bbf{b},r+1}\rb{x_1}=\varepsilon_{\bbf{b},r}\rb{x_1}+h s_{\bbf{b},r+1}\ind_{\left[\tau_{\bbf{b},r+1},\infty\right)}\rb{x_1},
\label{eq:eps_recurrent}
\end{gather}	
where $r\in\cb{0,\dots,R_{\bbf{b}}^+-1}$, $\varepsilon_{\bbf{b},r}$ is a recursive sequence of functions for computing the residual $\varepsilon_f$ such that $\varepsilon_{\bbf{b},0}\rb{x_1}=hM_{\bbf{b}}$,  $\varepsilon_{\bbf{b},R_{\bbf{b}}^+}\rb{x_1}=\varepsilon_f\rb{\mathbf{Vx}},\bbf{x}\in \cR_{\bbf{b}}$ and $\tau_{\bbf{b},0}=s_{\bbf{b},0}=0$. Furthermore $R_{\bbf{b}}^+\in\Z_+$ satisfies
\begin{equation}
    \cb{x_1>\tau_{\bbf{b},R_{\bbf{b}}^+}\setsep \sup_{\bbf{x}\in\cR_{\bbf{b}}} \vb{f_{\bbf{x}}\rb{x_1}-\varepsilon_{\bbf{b},R_{\bbf{b}}^+}\rb{x_1}}=\lambda}=\emptyset.
    \label{eq:Rmax}
\end{equation}
Furthermore, for $x_1<0$ we have \[\MOh^D f\rb{\mathbf{Vx}}=\MOh^D\sqb{f^-}\rb{-x_1\mathbf{v}_1+x_2\mathbf{v}_2+\dots+x_D\mathbf{v}_D},\] 
where $f^-\rb{\mathbf{x}}=f_{\bbf{x}}\rb{-x_1}$. Let $\tau_{\bbf{b},r}^-,s_{\bbf{b},r}^-,r\geq1$ be the folding times and signs computed via \eqref{eq:tau_recurrent}, \eqref{eq:sign_recurrent} for $\MOh^D\sqb{f^-}$. Then we define $\tau_r\triangleq -\tau^-_{-r},s_r\triangleq s^-_{-r}, r\in\Z,r<0$.
\label{def:cont_moduloMD}
\end{definition}

We note that the operator in Definition \ref{def:cont_moduloMD} is backwards compatible with the \D operator in Definition \ref{def:cont_modulo}. Specifically, if one chooses $\cR_{\bbf{b}}$ to be a single point $\cR_{\bbf{b}}=B \bbf{b}\in \RDred$ instead of a hypercube, then $\tau_{\bbf{b},r}$ and $s_{\bbf{b},r}$ in \eqref{eq:tau_recurrent} are the same as $\tau_r$ in Definition \ref{def:cont_modulo}. Furthermore, it was shown that, for $h=0$, the \D modulo-hysteresis operator in Definition \ref{def:cont_modulo} is identical to an ideal modulo operator \eqref{eq:ideal_mod} \cite{Florescu:2022:J}.

\begin{figure}
    \centering
    \begin{tabular}{c}
    \includegraphics[width=1\textwidth]{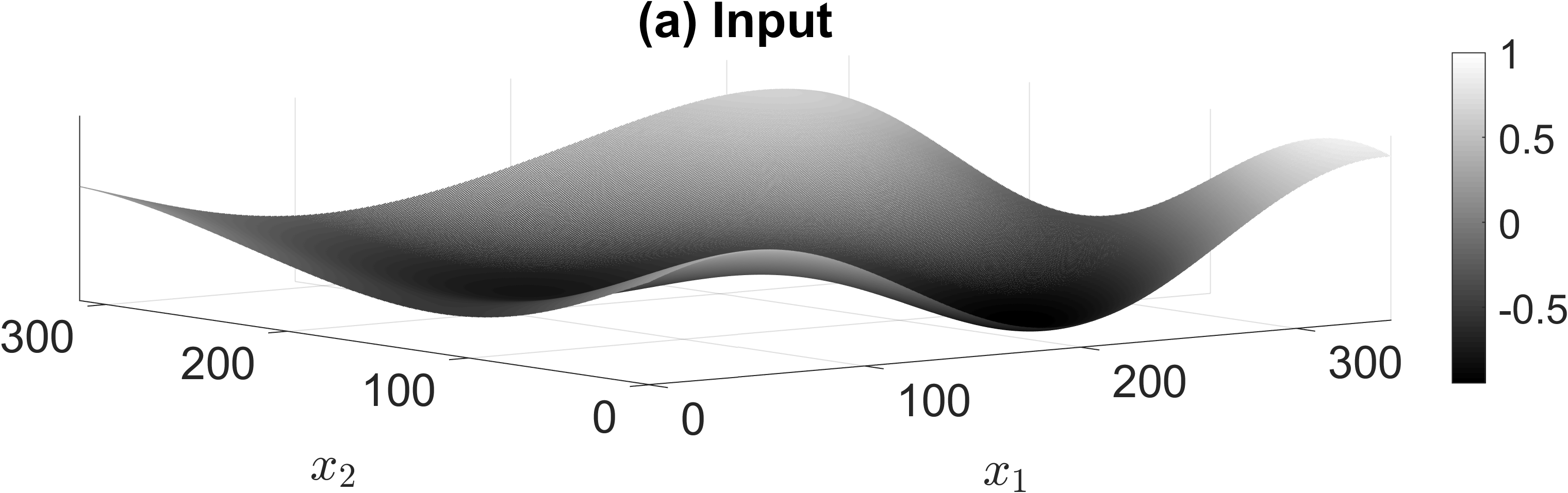}\\
    \begin{tabular}{cc}
        \includegraphics[width=0.44\textwidth]{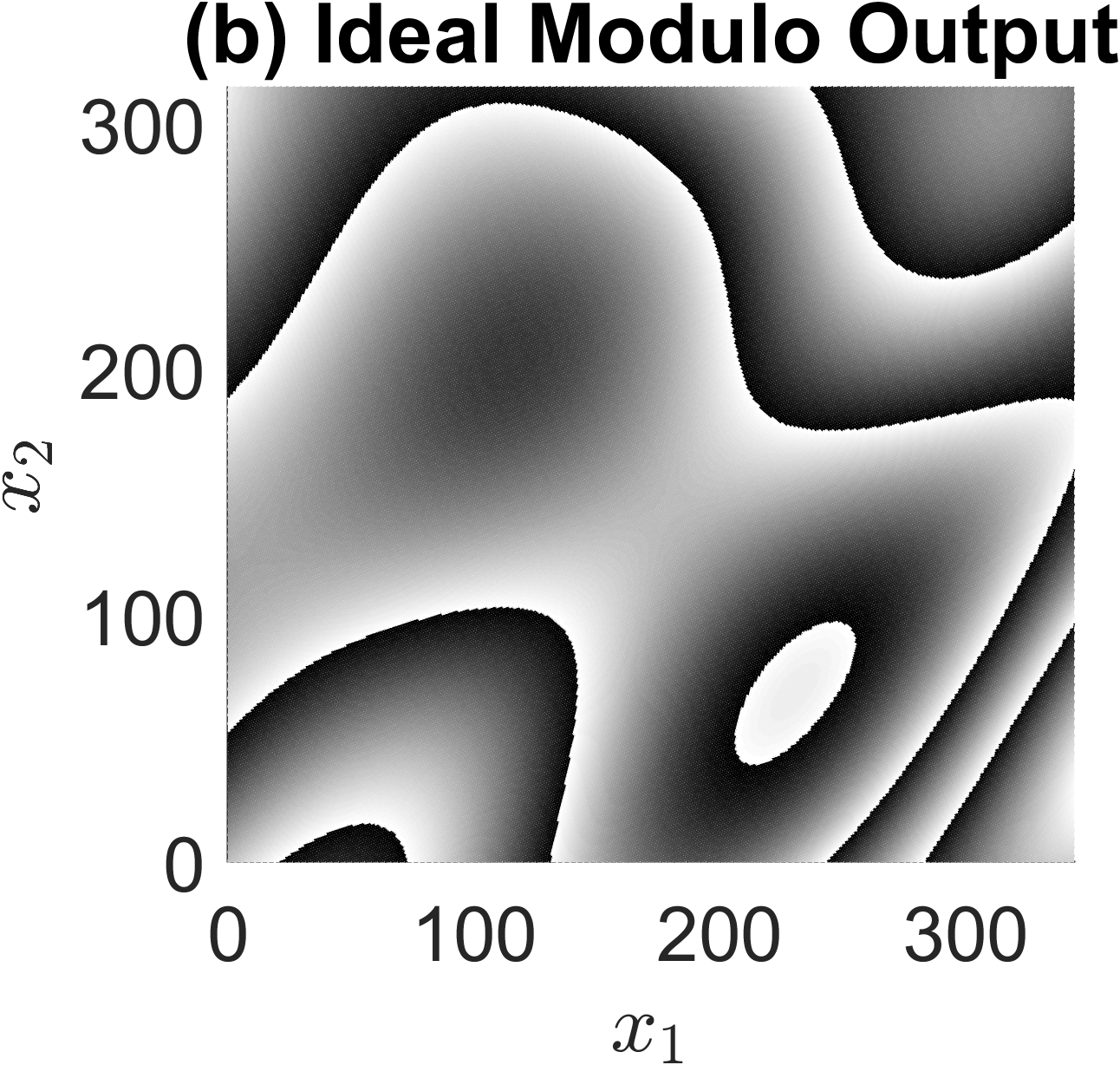} & \includegraphics[width=0.5\textwidth]{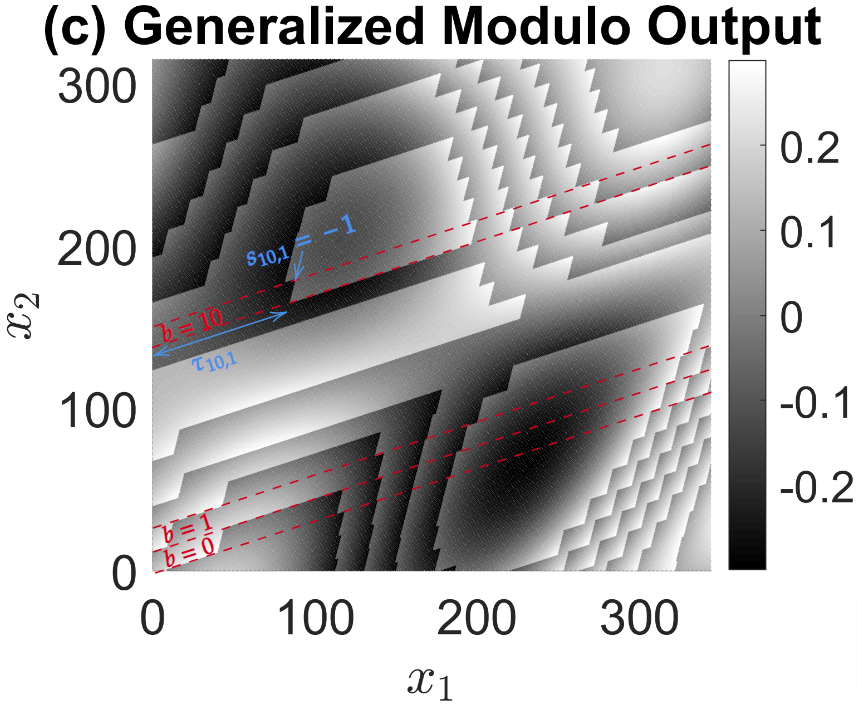} \\
        \includegraphics[width=0.44\textwidth]{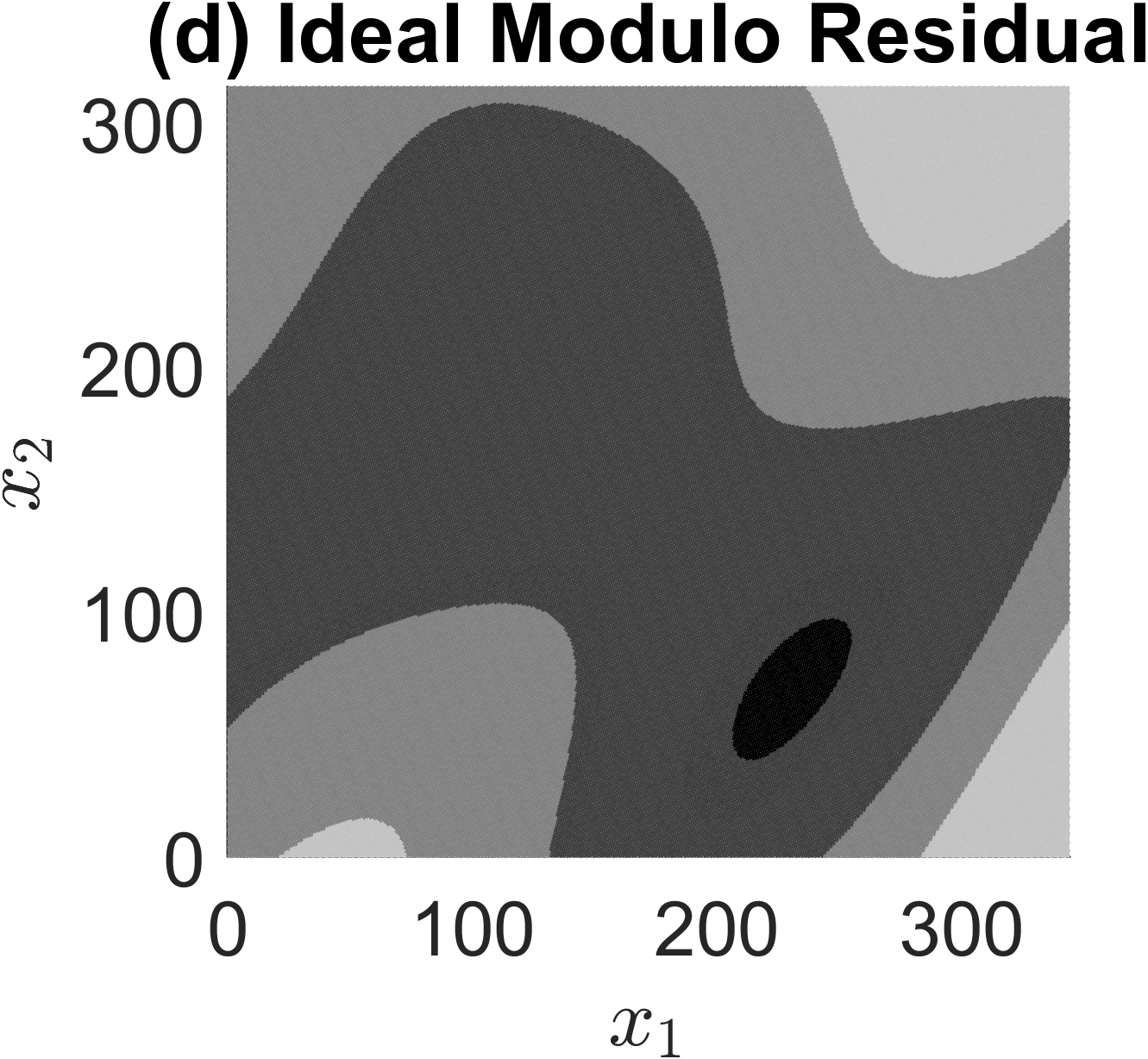} &
        \includegraphics[width=0.48\textwidth]{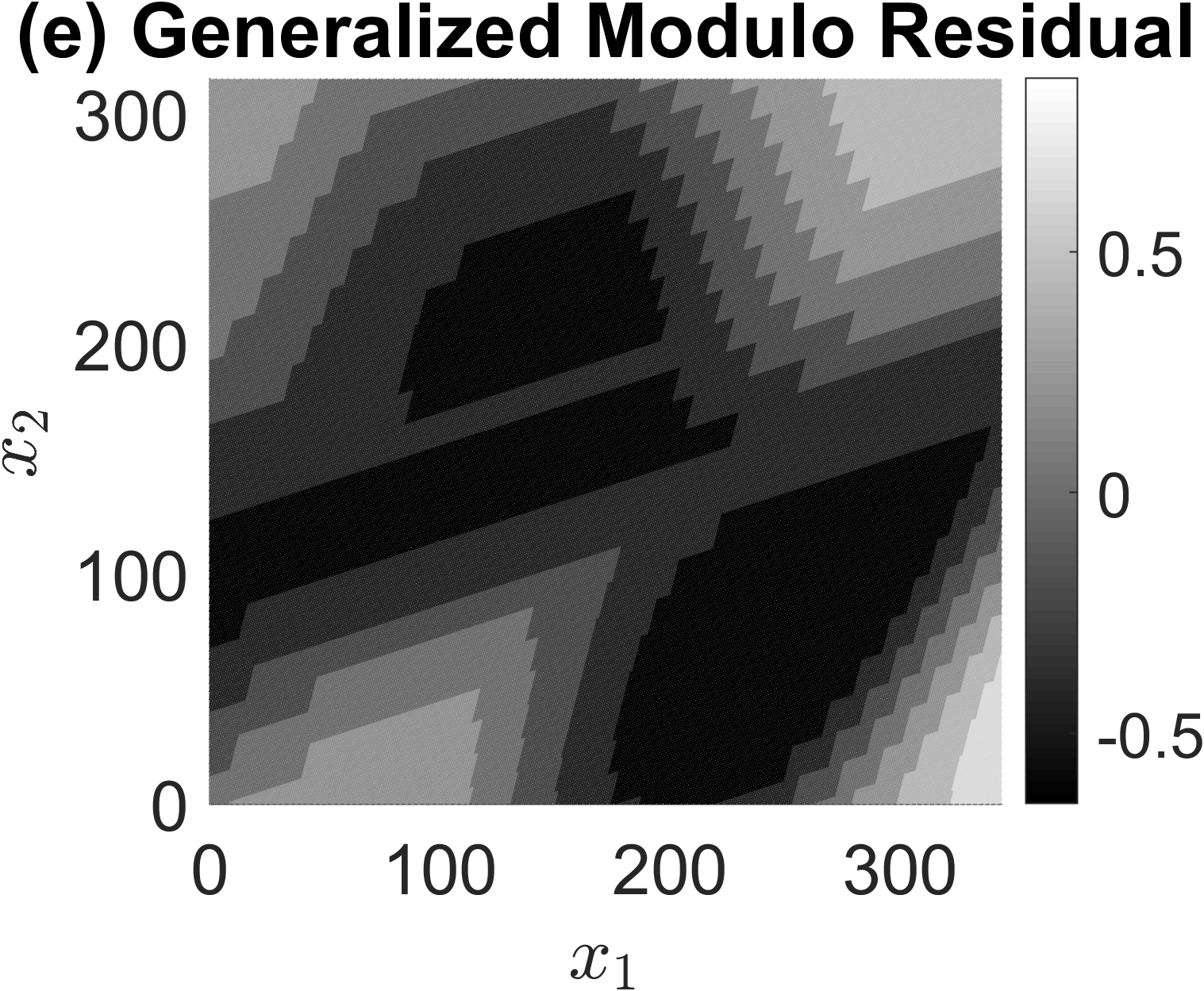}
    \end{tabular}
    \end{tabular}
    \caption{A random two-dimensional bandlimited input $f\rb{x_1,x_2}$ was generated (a). The ideal modulo output $\MO f$ is in (b) and the generalized modulo output $\MOh^D f$ in (c). The corresponding residual functions are depicted in (d) for ideal modulo and in (e) for generalized modulo. For $\MOh^D$, the lattice $\bs{\Lambda}$ consists of vectors $\mathbf{v}_1=\sqb{0.97,0.25}^\top, \mathbf{v}_2=\sqb{0.32,0.95}^\top$. The folding curves of the ideal modulo  -- the contours in (b) \& (d) -- are unknown \emph{a priori}. In the case of the modulo-hysteresis, folding occurs along straight lines with directions dictated by the lattice $\bs{\Lambda}$ which is known \emph{a priori}. This property will be exploited in recovery.}
    \label{fig:input}
\end{figure}

\subsection{Properties of the proposed operator}

In the following we give a number of properties of the \MD modulo-hysteresis operator for a bandlimited input.
\begin{prop}[Folding time separation]
Assume that $\tau_{\bbf{b},r}$ are well-defined in \eqref{eq:tau_recurrent} for $f\in\pwd$ and $r\in\cb{1,\dots,R}$ where $R\geq1$ and $\tau_{\bbf{b},0}=0$. Furthermore, assume that $\df <\min\cb{h/2,2\lambda-3h}$, where
\begin{equation}
\df\triangleq\mathop{\sup_{
    \bbf{b}\in\ZDred}}_{x_1\in\R} \sqb{\sup_{\bbf{x}\in \cR_{\bbf{b}}} f_{\bbf{x}}\rb{x_1}-\inf_{\bbf{x}\in \cR_{\bbf{b}}} f_{\bbf{x}}\rb{x_1}}.
\label{eq:DelfB}
\end{equation}
Then 
\begin{equation}
    \tau_{\bbf{b},r+1}-\tau_{\bbf{b},r}\geq \frac{h}{\Omega_1 \norm{f}_\infty},\quad \forall \bbf{b} \in \ZDred.
\end{equation}
\label{prop:fold_sep_multid}
\end{prop}
\if\FoldSepInAppendix\FL
    \input{Proofs/FoldSep}
\else
    \begin{proof}
    The proof is in Section \ref{sect:proofs_prop}.
    \end{proof}
\fi

\begin{prop}[Well-defined operator]
\label{prop:well_defined_op}
For input $f\in\pwd$ operator $\MOh^D$ in Definition \ref{def:cont_modulo} is \emph{well-defined} if $\df <\min\cb{h/2,2\lambda-3h}$, where $\df$ satisfies \eqref{eq:DelfB}.
\end{prop}
\if\WellDefInAppendix\FL
    \begin{proof}
    \input{Proofs/WellDef}
    \end{proof}
\else
    \begin{proof}
    The proof is in Section \ref{sect:proofs_prop}.
    \end{proof}
\fi

\begin{prop}[Modulo output dynamic range]
Let $\MOh^D$ be the \MD modulo-hysteresis operator in Definition \ref{def:cont_moduloMD} and $f\in\pwd$ such that $\df$ \eqref{eq:DelfB} satisfies $\df <\min\cb{h/2,2\lambda-3h}$. Then $\MOh^D f\rb{\mathbf{x}} \in \sqb{-\lambda,\lambda}, \forall \mathbf{x} \in \RD$.
\label{prop:mod_dyn_range}
\end{prop}
\if\ModDynRangeInAppendix\FL
    \input{Proofs/ModDynRange}
\else
    \begin{proof}
    The proof is in Section \ref{sect:proofs_prop}.
    \end{proof}
\fi

Proposition \ref{prop:mod_dyn_range} shows that operator $\MOh^D$ has a similar effect as the \D modulo nonlinearity, in that it keeps a signal within a fixed dynamic range $\sqb{-\lambda,\lambda}$.

\begin{cor}[Bound for intra-band variation]
\label{cor:bound_fB}
The quantity $\df$ defined in Proposition \ref{prop:fold_sep_multid} can be bounded as 
\begin{equation}
    \df\leq \norm{f}_\infty \cdot  B\sqrt{D} \cdot \norm{\bs{\Omega}}_2.
\end{equation}
\end{cor}
\if\BoundDelfBInAppendix\FL
    \begin{proof}
    \input{Proofs/BoundDelfB}
    \end{proof}
\else
    \begin{proof}
    The proof is in Section \ref{sect:proofs_prop}.
    \end{proof}
\fi

In \fig{fig:input} the variation of $M_{\bbf{b}}$, which can be seen as folds along coordinates $x_2,\dots,x_D$ for $x_1=0$, is \emph{gradual}, meaning that $M_{\bbf{b}}$ changes by $1$ between neighboring bands. Formally, we define by $\nbr$ the set comprising all neighboring bands of band $\bbf{b}$ below.
\begin{definition}[Neighboring bands]
\label{def:neighboring_bands}
    The set $\nbr$ of vectors neighboring $\bbf{b}\in\ZDred$ is defined as the set of all $\bbf{b}^*\in\ZDred$ for which $\exists d^*\in\cb{2,\dots,D}$ such that $\vb{\sqb{\bbf{b}^*}_{d^*}-\sqb{\bbf{b}}_{d^*}}=1$ and $\sqb{\bbf{b}}_d=\sqb{\bbf{b}^*}_d, \forall d\in\cb{2,\dots,D}\setminus d^*$.
\end{definition}
In the following, we provide conditions for which $M_{\bbf{b}^*}-M_{\bbf{b}} \in \cb{-1,0,1}$ where $\bbf{b}^*\in\nbr$.
\begin{prop}[Variation of $M_{\bbf{b}}$]
\label{prop:variation_Mb}
For $\forall\bbf{b}\in\ZDred$, let $\bbf{b}^*\in\nbr$. Let $M_{\bbf{b}^*}$ and $M_{\bbf{b}}$ be the modulo-hysteresis constants for an input $f\in\pwd$ satisfying the following condition as per Definition \ref{def:cont_moduloMD}
\begin{equation}
    \norm{f}_\infty B\sqrt{D}\cdot \sqrt{\sum_{d=1}^D \Omega_d^2}<\min \cb{\frac{h}{2},2\lambda-3h}.
    \label{eq:suff_cond_Delf}
\end{equation}
Then $M_{\bbf{b}^*}-M_{\bbf{b}} \in \cb{-1,0,1}$.
\end{prop}
\begin{proof}
We first note that sets $\cR_{\bbf{b}}$ and $\cR_{\bbf{b}^*}$ are neighboring polytopes which, via their definition, satisfy $\mathrm{cl}\rb{\cR_{\bbf{b}}}\cap \mathrm{cl}\rb{\cR_{\bbf{b}^*}}\neq \emptyset$. Using this in conjunction with the properties of supremum and infimum, we get
\begin{equation}
\label{eq:sup_def_via_max}
    \sup_{\bbf{x}\in\cR_{\bbf{b}}} f_{\bbf{x}}\rb{0}
    =\max_{\bbf{x}\in{\mathrm{cl}\rb{\cR_{\bbf{b}}}}} f_{\bbf{x}}\rb{0}
    \geq\min_{\bbf{x}\in{\mathrm{cl}\rb{\cR_{\bbf{b}^*}}}} f_{\bbf{x}}\rb{0}
    =\inf_{\bbf{x}\in\cR_{\bbf{b}^*}} f_{\bbf{x}}\rb{0}.
\end{equation}
Condition \eqref{eq:suff_cond_Delf} implies $\df<h/2$ due to Corollary \ref{cor:bound_fB}. Therefore,
\begin{equation}
\label{eq:particular_Del_fB_bound}
    \sup_{\bbf{x}\in\cR_{\bbf{b}}} f_{\bbf{x}}\rb{0}-\inf_{\bbf{x}\in\cR_{\bbf{b}}} f_{\bbf{x}}\rb{0}<h/2.
\end{equation}
Using \eqref{eq:sup_def_via_max} and \eqref{eq:particular_Del_fB_bound}
\begin{equation}
    \inf_{\bbf{x}\in\cR_{\bbf{b}^*}} f_{\bbf{x}}\rb{0}-\inf_{\bbf{x}\in\cR_{\bbf{b}}} f_{\bbf{x}}\rb{0}<h/2.    
\end{equation}
Given that, by definition, $    M_{\bbf{b}^*}=\floor{\frac{\inf_{\bbf{x}\in \cR_{\bbf{b}^*}} f_{\bbf{x}}\rb{0}+\lambda}{h}}-1$, it can be shown by direct derivation that $M_{\bbf{b}^*}\leq M_{\bbf{b}}+1$. Furthermore, by swapping $\bbf{b}$ and $\bbf{b}^*$ in the derivation above we get
\begin{equation}
    M_{\bbf{b}^*}\in\cb{M_{\bbf{b}}-1,M_{\bbf{b}},M_{\bbf{b}}+1}.
\end{equation}
\end{proof}
Therefore, using Definition \ref{def:cont_moduloMD}, for $x_1=0$, in neighboring bands $\bbf{b}$ and $\bbf{b}^*$ the residual $\varepsilon_f$ is either the same, or differs by $h$. This is similar to the behavior of the residual around a folding time $\tau_{\bbf{b},r}$ along dimension $x_1$. 

\subsection{Problem formulation}

The measurements $y\sqb{\mathbf{k}}$ are assumed to be samples on a \MD lattice $\bs{\Lambda}=\cb{\mathbf{VTk}\setsep \mathbf{k}\in\ZD}$, such that
\begin{equation}
y\sqb{\mathbf{k}}=\MOh^D f(\mathbf{VTk})+\eta\sqb{\mathbf{k}}=\gamma\sqb{\mathbf{k}}-\varepsilon_{\gamma}\sqb{\mathbf{k}}+\eta\sqb{\mathbf{k}},    
\end{equation}
where $\gamma\sqb{\mathbf{k}}=f\rb{\mathbf{VTk}}$, $\varepsilon_{\gamma}\sqb{\mathbf{k}}=\varepsilon_f\rb{\mathbf{VTk}}$ and $\eta\sqb{\mathbf{k}}\sim \mathcal{N}\rb{0,\sigma^2}$.
The known variables are the input bandwidth $\bs{\Omega}$, number of dimensions $D$, the lattice $\bs{\Lambda}$, modulo-hysteresis parameters $B,\lambda,h$ and output samples $y\sqb{\mathbf{k}}$. The proposed reconstruction problem is to compute the input lattice samples $\widetilde{\gamma}\sqb{\mathbf{k}}$ defined as
\begin{equation}
    \widetilde{\gamma}\sqb{\mathbf{k}}= \gamma\sqb{\mathbf{k}}+\eta\sqb{\mathbf{k}}+ Mh,
    \label{eq:up_to_constant}
\end{equation}
where $M\in\Z$ is an unknown integer. The input $\gamma\sqb{\mathbf{k}}$ can only be reconstructed up to an integer multiple of $h$ given that $\MOh^D f\rb{\mathbf{x}}=\MOh^D \sqb{f+Mh}\rb{\mathbf{x}},$ $ \forall M\in\Z$ (\ref{eq:zt_MD},\ref{eq:eps}). 

\section{Detecting modulo-hysteresis discontinuities}
\label{sect:det_disc}
We define $N_d^B=\frac{B}{T_d}$. For simplicity, we assume that $N_d^B\in\Z, N_d^B\geq1,\forall d\in\cb{2,\dots D}$, i.e., for a polytope  $\cP_{\bbf{b}}$ we have an integer number of sampling periods $T_d$ along each of its edges of length $B$ and direction $\mathbf{v}_d$. The recovery is achieved in several steps
\begin{enumerate}
    \item Compute $M_{\bbf{b}}$, the folding times $\tau_{\bbf{b},r}$ and signs $s_{\bbf{b},r}$ in each band $\cB_{\bbf{b}}$.
    \item Compute residual $\varepsilon_f\rb{\mathbf{Vx}}$ \eqref{eq:eps}.
    \item Compute the samples $\widetilde{\gamma}\sqb{\mathbf{k}}$. 
\end{enumerate}

For step 1, we define a filter $\psi$ and compute $\inner{y,\psi}$. For detecting the folding times and signs, we choose $\psi=\psi_{\bbf{b},m}$, which is centered in sample $m$ along dimension $x_1$ in band $\bbf{b}$ as 
\begin{equation}
    \psi_{\bbf{b},m}\sqb{\mathbf{k}}=
    \left\lbrace \begin{array}{cc}
        \Delta^N\sqb{k_1-m}\cdot \frac{1}{N^B}, & \bbf{T}\bbf{k}\in\cR_{\bbf{b}}, \\
        0, & \bbf{T}\bbf{k}\notin\cR_{\bbf{b}},
    \end{array}\right.
    \label{eq:psibm}
\end{equation}
where $\bbf{k}\in\ZDred$ and $N^B=\prod_{d=2}^D N_d^B$ is the number of samples in each band $\bbf{b}$ for $k_1$ fixed. For fixed $\bbf{k}\in\ZDred$, $\psi_{\bbf{b},m}\sqb{\mathbf{k}}$ is a finite difference filter along dimension $k_1$.

For detecting the $M_{\bbf{b}}$, we need a filter detecting the change in $M_{\bbf{b}}$  for two bands $\bbf{b}$, $\bbf{b}^*$, such that $\bbf{b}^*\in\nbr$ and $\sqb{\bbf{b}^*}_{d^*}=\sqb{\bbf{b}}_{d^*}+1$. We define $\psi=\psi_{\bbf{b},\bbf{b}^*}\sqb{\mathbf{k}}$ as
\begin{equation}
    \psi_{\bbf{b},\bbf{b}^*}\sqb{\mathbf{k}}=
    \left\lbrace \begin{array}{cc}
        \Delta^N\sqb{k_{d^*}-k_{\bbf{b}}}\cdot \frac{1}{N^{B*}}, & \bbf{T}\bbf{k}\in\cR_{\bbf{b}}\cup\cR_{\bbf{b}^*}, k_1=0, \\
        0, & \quad \mathrm{otherwise},
    \end{array}\right.
    \label{eq:psibb}
\end{equation}
where $k_{\bbf{b}}=N_{d^*}^B\cdot\sqb{\bbf{b}}_{d^*}-1$ and $N^{B*}=\prod_{d=2,d\neq d^*}^D N_d^B$ is the number of samples in each band $\bbf{b}$ for $k_1$, $k_{d^*}$ fixed.  Therefore, similar to $\psi_{\bbf{b},m}$, filter $\psi_{\bbf{b},\bbf{b}^*}$ is a finite difference filter along dimension $x_{d^*}$ which is perpendicular to the hyperplane separating bands $\bbf{b}$ and $\bbf{b}^*$. Furthermore, $\psi_{\bbf{b},\bbf{b}^*}$ is constant within each band $\bbf{b}$ and $\bbf{b}^*$. This means that the finite difference filter is repeated $N^{B*}$ times across dimensions $d\in\cb{2,\dots,D}\setminus d^*$, which has a noise averaging effect.

\subsection{Detecting folding times and signs}

To detect $\tau_{\bbf{b},r}$ and $s_{\bbf{b},r}$, we use filter $\psi_{\bbf{b},m}$ to compute sequence $y_{\bbf{b},m}$
\begin{equation}
    y_{\bbf{b},m}=\inner{y,\psi_{\bbf{b},m}} =\underbrace{\inner{\gamma,\psi_{\bbf{b},m}}}_{\textsf{Input}}-\underbrace{\inner{\varepsilon_{\gamma},\psi_{\bbf{b},m}}}_{\textsf{Residual}}+
    \underbrace{\inner{\eta,\psi_{\bbf{b},m}}}_{\textsf{Noise}}.
    \label{eq:filtered_data}
\end{equation}    
In \eqref{eq:filtered_data} the filtered samples $y_{\bbf{b},m}$ are composed of three terms: the \emph{input term} $\inner{\gamma,\psi_{\bbf{b},m}}$, the \emph{residual term} $\inner{\varepsilon_{\gamma},\psi_{\bbf{b},m}}$ and the \emph{noise term} $\inner{\eta,\psi_{\bbf{b},m}}$. 
Given that all three are unknown, the general recovery strategy is to separate them via thresholding; as will be shown later, thresholding samples $y_{\bbf{b},m}$ allows to compute the folding times and signs. While this will be derived rigorously later in propositions \ref{prop:fold_det} and \ref{prop:Mb_det}, here we give a brief intuitive explanation of the recovery method, by explaining the effect that filter $\psi_{\bbf{b},m}$ has on all $3$ terms in the right-hand-side of \eqref{eq:filtered_data}. A similar analysis applies to filter $\psi_{\bbf{b},\bbf{b}^*}$ which will be described in Section \ref{sect:Mb}.
    
As noted before, $\psi_{\bbf{b},m}$ is a finite difference filter along dimension $x_1$. It was shown for the \D case that this causes $\gamma\sqb{\mathbf{k}}$ to vanish for large $N$. Furthermore, it generates peaks at the folding times in residual $\varepsilon_\gamma$ and also amplifies the noise $\eta$ \cite{Bhandari:2020:Ja,Florescu:2022:J}. This latter effect is undesirable for recovery. To decrease the effect of the noise the finite difference filter was convolved with a spline in the \D case, but this also makes the detection of $\varepsilon_\gamma$ more difficult \cite{Florescu:2022:Cb}. Here we can address noise filtering without affecting the folding time detection by exploiting the \MD structure of $\MOh^D$.
    
For fixed $k_1\in\Z$, the filter $\psi_{\bbf{b},m}\sqb{\mathbf{k}}$ is constant along dimensions $k_2,\dots,k_D$ as long as $\bbf{V}\bbf{T}\bbf{k}\in\cR_{\bbf{b}}$ and thus the inner product has an averaging effect. However, we know that the modulo residual corresponding to $\MOh^D f\rb{\mathbf{Vx}},\bbf{x}\in\cR_{\bbf{b}}$ is also constant within $\cR_{\bbf{b}}$ by definition, and therefore the averaging effect does not affect the residual edges, which are along dimension $k_1$. Moreover, given that $f$ is smooth and changes slowly within a band $\cB_{\bbf{b}}$, the filter averaging along dimensions $\bbf{k}$ has very little effect on $\gamma$. Therefore, along dimensions $\bbf{k}$, the filter acts mainly on the noise sequence $\eta\sqb{\bbf{k}}$ by narrowing its \pdf around the origin such that its effect gradually vanishes.

As in the \D case, the modulo output $z$ is smooth in-between the folds, and has discontinuities at the folding times. The filter $\psi_{\bbf{b},m}$ responds with pulses of non-zero support to the input discontinuities. To account for this, we define by $\mathbb{S}_N$ the support of the filtered residual such that (see \cite{Florescu:2022:J} for details)
\[
\mathbb{S}_N=\mathrm{supp}\sqb{\inner{\varepsilon_{\gamma},\psi_{\bbf{b},m}}}=\underset{r\in\Z^*}{\bigcup}\cb{\floor{\frac{\tau_{\bbf{b},r}}{T_1}}-N,\dots,\floor{\frac{\tau_{\bbf{b},r}}{T_1}}}.
\]

The following theorem shows how $\mathbb{S}_N$ can be recovered by thresholding sequence $\inner{y,\psi_{\bbf{b},m}}$, which is an important step in computing folding times $\tau_{\bbf{b},r}$.
\begin{prop}[Detection of folding times]
\label{prop:fold_det}
Let $f\in\pwd$ and let $z\rb{\mathbf{x}}=\MOh^D f\rb{\mathbf{x}}$ be the output of a \MD modulo-hysteresis model with parameters $\lambda,h,B$. Furthermore, let $y\sqb{\mathbf{k}}=z\rb{\mathbf{VTk}}+\eta\sqb{\mathbf{k}}$ be the samples of the modulo output computed on lattice $\bs{\Lambda}$ corrupted by a noise sequence $\eta\sqb{\mathbf{k}}\sim\mathcal{N}\rb{0,\sigma^2}$. Furthermore, assume that $T_d<B,\forall d \in\cb{2,\dots,D}$ and that
\begin{align}
\begin{split}
    \norm{f}_\infty B\sqrt{D}\cdot \sqrt{\sum_{d=1}^D \Omega_d^2}&<\min \cb{\frac{h}{2},2\lambda-3h},\\
    \rb{T_1\Omega_1 e}^N\norm{f}_\infty&<h/2.
\end{split}
\label{eq:rec_cond}
\end{align}
If $\vb{\inner{y,\psi_{\bbf{b},m}}}\geq h/2$ then $m\in \mathbb{S}_N$ with probability $p>1-p_{\mathsf{err}}$ where
\begin{equation}
    p_{\mathsf{err}}\leq e^{-C^2}, \quad \mathrm{where}\ C=\frac{h/2-\rb{T_1\Omega_1 e}^N\norm{f}_\infty}{\sigma \sqrt{2^{N+1}}}\prod_{d=2}^D\sqrt{\frac{B}{T_d}}.
    \label{eq:perr0}
\end{equation}
\end{prop}
\if\FoldDetInAppendix\FL
    \begin{proof}
    \input{Proofs/FoldDet}

    \end{proof}
\else
    \begin{proof}
    The proof is in Section \ref{sect:proofs_rec}.
    \end{proof}
\fi

Due to Proposition \ref{prop:fold_det}, for $\forall \sigma,\lambda,h,\bs{\Omega}\in\R^{D}_+$ satisfying \eqref{eq:rec_cond} and a fixed $m\in\Z$, one can choose $T_1,\dots,T_D>0$ such that the truth value of $m\in \mathbb{S}_N$ is evaluated correctly with an arbitrarily large probability. We note that $p_{\mathsf{err}}$ measures the probability when a recovery error is possible, but not guaranteed, therefore the error probability is smaller in a real scenario.
A small error in Proposition \ref{prop:fold_det} means a large $C$, which can be achieved by decreasing the sampling periods $T_1,\dots,T_d$ or increasing $B,h$ or number of dimensions $D$. 

The residual $\varepsilon_\gamma$, used for reconstructing $\gamma$, requires detecting constants $M_{\bbf{b}}$ in addition to the folding times ${\tau}_{\bbf{b},r}$ and signs ${s}_{\bbf{b},r}$, as will be explained in the next subsection. 

\subsection{Detecting constants $M_{\bbf{b}}$}
\label{sect:Mb}
Given that we can only recover $\gamma$ up to an integer multiple of $h$ \eqref{eq:up_to_constant}, we define $\widetilde{M}_{\bbf{b}}\triangleq M_{\bbf{b}}-M_{\bbf{0}}$, where $\bbf{0}$ is the null vector of $\ZDred$, and recover $\widetilde{M}_{\bbf{b}}$. 
Just as the folding times, different values of $M_{\bbf{b}}$ for adjacent bands cause discontinuities. However, unlike the detection of the folding times, here we have additional information. Specifically, we know that the discontinuities may only be located at the neighboring sides of polytopes $\cP_{\bbf{b}}$. We use  the filter $\psi_{\bbf{b},\bbf{b}^*}$ \eqref{eq:psibb} to detect the discontinuities in a similar fashion to detecting the folding times via $\psi_{\bbf{b},m}$. This time, however, the finite differences $\Delta^N$ computed via $\psi_{\bbf{b},\bbf{b}^*}$ evaluate variations across dimensions $x_2,\dots,x_D$. 

\begin{prop}[Detection of constants $M_{\bbf{b}}$]
\label{prop:Mb_det}
Let $y\sqb{\mathbf{k}}=z\rb{\mathbf{VTk}}+\eta\sqb{\mathbf{k}}$, where $\eta\sqb{\mathbf{k}}\sim\mathcal{N}\rb{0,\sigma^2}$ and $z\rb{\mathbf{x}}$ is the output of a modulo-hysteresis operator
$z\rb{\mathbf{x}}=\MOh^D f\rb{\mathbf{x}}$ with parameters $\lambda,h$ and $f\in\pwd$ satisfying
\begin{equation}
\norm{f}_\infty B\sqrt{D}\cdot \sqrt{\sum_{d=1}^D \Omega_d^2}<\min \cb{h/2,2\lambda-3h},
\label{eq:rec_cond2}
\end{equation}
Furthermore, let $\bbf{b}\in\ZDred$, $\bbf{b}^*\in\nbr$ and $d^*\in\cb{2,\dots,D}$ such that $\sqb{\bbf{b}^*}_{d^*}=\sqb{\bbf{b}}_{d^*}+1$. Assume that 
\begin{gather}
\rb{T_{d^*}\Omega_{d^*} e}^N\norm{f}_\infty<h/2,\\
\rb{N+1}T_{d^*}<B.
\end{gather}
Then the following is true with probability $p>1-p_{\mathsf{err}}$
\begin{equation}
\begin{cases}
M_{\bbf{b}^*} =M_{\bbf{b}}+\mathrm{sign} \rb{\inner{y,\psi_{\bbf{b},\bbf{b}^*}}} & \mathrm{if} \vb{\inner{y,\psi_{\bbf{b},\bbf{b}^*}}} \geq h/2, \\
M_{\bbf{b}^*}=M_{\bbf{b}}   & \mathrm{otherwise}.
        \end{cases}
\end{equation}   
where
\begin{equation}
    p_{\mathsf{err}}\leq e^{-\kappa^2}, \quad \mathrm{where}\ \kappa=\frac{h/2-\rb{T_{d^*}\Omega_{d^*} e}^N\norm{f}_\infty}{\sigma \sqrt{2^{N+1}}}\prod_{d=2,d\neq d^*}^D\sqrt{\frac{B}{T_d}}.
\end{equation}
\end{prop}

\if\MbDetInAppendix\FL
    \begin{proof}
    \input{Proofs/MbDet}
    \end{proof}
\else
    \begin{proof}
    The proof is in Section \ref{sect:proofs_rec}.
    \end{proof}
\fi

\section{Input reconstruction}
\label{sect:inp_rec}
\subsection{Recovery with the proposed operator}

We begin with the noiseless input recovery scenario $\sigma=0$ where, via Proposition \ref{prop:fold_det}, the set $\mathbb{S}_N$ is perfectly identified with probability $1$. Furthermore, a constant $\widetilde{M}_{\bbf{b}}=M_{\bbf{b}}-M_{\bbf{0}}$ can be perfectly recovered from $\widetilde{M}_{\bbf{b}^*}$ where $\bbf{b}^*\in\nbr$ according to Proposition \ref{prop:Mb_det}. The following theorem proves the input recovery conditions in the case $\sigma=0$.
\begin{theo}[Noiseless input reconstruction]
\label{th:input_rec}
Let $f\in\pwd$ and let $z\rb{\mathbf{x}}=\MOh^D f\rb{\mathbf{x}}$ be the output of a \MD modulo-hysteresis model with parameters $\lambda,h,B$. Furthermore, let $y\sqb{\mathbf{k}}=z\rb{\mathbf{VTk}}$ be the samples of the modulo output computed on lattice $\bs{\Lambda}$. Furthermore, for $N\in\Z, N\geq1$, assume that
\begin{gather}
    \norm{f}_\infty B\sqrt{D}\cdot \norm{\bs{\Omega}}_2<\min \cb{h/2,2\lambda-3h},\label{eq:rec_cond1_input}\\
    \rb{T_d\Omega_d e}^N\norm{f}_\infty<h/2,\quad \forall d\in\cb{1,\dots,D},\label{eq:rec_cond2_input}\\    
    \rb{N+1}T_1<\frac{h}{\Omega_1 \norm{f}_\infty},\label{eq:rec_cond3_input}\\
    \rb{N+1}T_d<B,\quad\forall d\in\cb{2,\dots,D}.
    \label{eq:rec_cond4_input}
\end{gather}
Then samples $\widetilde{\gamma}\sqb{\mathbf{k}}=\gamma\sqb{\mathbf{k}}-hM_{\bbf{0}}$ can be perfectly reconstructed from $y\sqb{\mathbf{k}}$.
\end{theo}
\if\InputRecCleanInAppendix\FL
    \begin{proof}
    \input{Proofs/InputRecClean}
    \end{proof}
\else
    \begin{proof}
    The proof is in Section \ref{sect:proofs_rec}.
    \end{proof}
\fi

The interpretation of the sufficient conditions in Theorem \ref{th:input_rec} is as follows. The modulo-hysteresis is well-defined due to a bounded intra-band variation guaranteed by \eqref{eq:rec_cond1_input}. Condition \eqref{eq:rec_cond2_input} bounds the $N$-th order difference of the input along all of the dimensions, ensuring that the filter has enough shrinking effect on the input. Finally, \eqref{eq:rec_cond3_input} and \eqref{eq:rec_cond4_input} guarantee enough samples in between the folds \eqref{eq:rec_cond3_input} and within each band \eqref{eq:rec_cond4_input} so that the supports of the filters detecting consecutive discontinuities don't overlap.

In the general case where $\sigma>0$ the following result holds true.
\begin{theo}[Noisy input reconstruction]
\label{th:noisy_input_rec}
Let $f\in\pwd$ and $z\rb{\mathbf{x}}=\MOh^D f\rb{\mathbf{x}}$. Assume that $y\sqb{\mathbf{k}}=z\rb{\mathbf{VTk}}+\eta\sqb{\mathbf{k}}$ are known for $k_d\in\cb{1,\dots,K^d_{\mathsf{max}}}$ where $\eta\sqb{\mathbf{k}}\sim\mathcal{N}\rb{0,\sigma^2}$, such that  $B_{\mathsf{max}}^d\triangleq \frac{K^d_{\mathsf{max}}}{N_d^B} \in \Z$ for $d\in\cb{2,\dots,D}$. Then, if (\ref{eq:rec_cond1_input}-\ref{eq:rec_cond4_input}) are true, then the input samples $\widetilde{\gamma}\sqb{\mathbf{k}}$ can be recovered from $y\sqb{\mathbf{k}}$ with a probability $p>p_{\mathsf{acc}}$ such that
\begin{equation*}
    p_{\mathsf{acc}}\geq \rb{1-e^{-C^2}}^{K_{\mathsf{max}}^1\cdot \prod_{d=2}^D B_{\mathsf{max}}^d}\hspace{-0.2em}\cdot \rb{1-e^{-\kappa_{\mathsf{min}}^2}}^{\prod_{d=2}^D B_{\mathsf{max}}^d}\hspace{-0.2em},
\end{equation*}
where
\begin{align*}
    C&=\frac{h/2-\rb{T_1\Omega_1 e}^N\norm{f}_\infty}{\sigma \sqrt{2^{N+1}}}\prod_{d=2}^D\sqrt{\frac{B}{T_d}},\\ \kappa_{\mathsf{min}}&=\frac{h/2-\sqb{\rb{\max_{d}T_d\Omega_d} \cdot e}^N\norm{f}_\infty}{\sigma \sqrt{2^{N+1}}}\sqrt{\frac{B^{D-2}}{T_{\mathsf{max}}^{D-2}}},
\end{align*}
where $T_{\mathsf{max}}=\max_{d\in\cb{2,\dots,D}} T_d$.
\end{theo}
\begin{proof}
Theorem \ref{th:input_rec} assumes that Proposition \ref{prop:fold_det} and \ref{prop:Mb_det} hold with $p_{\mathsf{err}}=0$ for all filters $\psi_{\bbf{b},m}$ and $\psi_{\bbf{b},\bbf{b}^*}$. To calculate the overall error probability when this assumption is not true, we count the filters above, when used in reconstruction, as follows. There are a total of $\prod_{d=2}^D B_{\mathsf{max}}^d$ bands, and $K_{\mathsf{max}}^1$ samples along dimension $x_1$. Then the probability that Proposition \ref{prop:fold_det} holds for all filters $\psi_{\bbf{b},m}$ is $\rb{1-e^{-C^2}}^{K_{\mathsf{max}}^1\cdot \prod_{d=2}^D B_{\mathsf{max}}^d}\hspace{-0.2em}$.

Next, in the case of Proposition \ref{prop:Mb_det}, we bound the error probability as follows
\begin{gather}
    p_{\mathsf{err}}\leq e^{-\kappa^2}\leq e^{-\kappa_{\mathsf{min}}^2},\quad 
\kappa_{\mathsf{min}}=\tfrac{h/2-\sqb{\rb{\max_{d\in\cb{2,\dots,D}}T_d\Omega_d}e}^N\norm{f}_\infty}{\sigma\sqrt{2^{N+1}}}\cdot \sqrt{\tfrac{B^{D-2}}{T_{\mathsf{max}}^{D-2}}}.
\end{gather}
We note that we do not use all filters $\psi_{\bbf{b},\bbf{b}^*}$. Given that we use a set of samples that is contiguous along all dimensions, any band $\bbf{b}$ containing samples has at least one neighboring band $\bbf{b}^*\in\nbr$ that contains samples. Then, each constant $\widetilde{M}_{\bbf{b}}=M_{\bbf{b}}-M_{\bbf{0}}$ can be computed using a single evaluation of $\psi_{\bbf{b},\bbf{b}^*}$, which, in total, is evaluated $\prod_{d=2}^D B_{\mathsf{max}}^d-1$ times, and the theorem follows.
\end{proof}

\subsection{Comparison to ideal modulo recovery with Gaussian noise measurements}

The USF was not analysed in the presence of Gaussian noise, but rather on bounded noise \cite{Bhandari:2020:Ja,Bouis:2020:C}. However, USF can still be applied for recovery in the context of Gaussian noise, and the reconstruction would still be accurate in the instances when the noise sample with maximum amplitude satisfies the USF conditions. The modulo operator with Gaussian noise was considered before, but the main objective was denoising, rather than input reconstruction \cite{Fanuel:2021,Tyagi:2022:J}. In order to assess the advantage of the new $\MOh^D$ operator, we provide some insight on reconstruction via USF for \MD inputs. Specifically, we note that, for an input $f\in\pwd$ and a lattice $\bs{\Lambda}$, the ideal modulo output is decomposed as (see also Section \ref{subsect:USF})
\[
y\sqb{\mathbf{k}}=\gamma\sqb{\mathbf{k}}-\varepsilon_{\gamma}\sqb{\mathbf{k}}+\eta\sqb{\mathbf{k}}.
\]
The recovery method from \cite{Bouis:2020:C} involves a \emph{line-by-line} approach, meaning that the recovery is performed along dimension $k_1$, $\forall \bbf{k}=\sqb{k_2,\dots,k_D}$. By defining $y_{\bbf{k}}\sqb{k_1}=y\sqb{\mathbf{k}}$, $\gamma_{\bbf{k}}\sqb{k_1}=\gamma\sqb{\mathbf{k}}$, $\gamma_{\bbf{k}}\sqb{k_1}=\gamma\sqb{\mathbf{k}}$, $\varepsilon_{\gamma,\bbf{k}}\sqb{{k_1}}=\varepsilon_{\gamma}\sqb{\mathbf{k}}$, and $\eta_{\bbf{k}}\sqb{{k_1}}=\eta\sqb{\mathbf{k}}$, the recovery is performed by computing
\begin{align}
\begin{split}
    \inner{y_{\bbf{k}},\Delta^N\sqb{\cdot-m}}&=\inner{\gamma_{\bbf{k}},\Delta^N\sqb{\cdot-m}}-\inner{\varepsilon_{\gamma,\bbf{k}},\Delta^N\sqb{\cdot-m}}\\
    &+\inner{\eta_{\bbf{k}},\Delta^N\sqb{\cdot-m}}.
\end{split}
\label{eq:USF_filtering}
\end{align}
We remark that the processing in \eqref{eq:USF_filtering} is equivalent to applying filter $\psi_{\bbf{b},m}$ \eqref{eq:psibm} in the case of the multidimensional modulo-hysteresis operator when there is only one sample per band in all dimensions, i.e., $N_2^B=N_3^B=\dots=N_d^B=N^B=1$. 
Even though the noise here is not bounded, we can derive the condition when the USF would work for a specific noise instance, which is  \cite{Bhandari:2020:Ja} \[\vb{\inner{\gamma_{\bbf{k}},\Delta^N\sqb{\cdot-m}}+\inner{\eta_{\bbf{k}},\Delta^N\sqb{\cdot-m}}}\leq \rb{T_1\Omega_1 e}^N\norm{f}_\infty+\vb{\inner{\eta_{\bbf{k}},\Delta^N\sqb{\cdot-m}}}<\lambda.\]
for all $m \in \Z$. Thus, recovery is only guaranteed if the noise instance is bounded by
\begin{equation}
\vb{\inner{\eta_{\bbf{k}},\Delta^N\sqb{\cdot-m}}}<\lambda-\rb{T_1\Omega_1 e}^N\norm{f}_\infty.
\label{eq:USF_error}
\end{equation}
In a similar fashion to the derivation of \eqref{eq:noise3} it can be shown that $\inner{\eta_{\bbf{k}},\Delta^N\sqb{\cdot-m}}\sim \mathcal{N}\rb{0,\sigma^2\cdot 2^N}$. We remark that the standard deviation of $\inner{\eta_{\bbf{k}},\Delta^N\sqb{\cdot-m}}$ is always at least $\sigma\sqrt{2}$. Therefore, depending on the values of $T_1, \Omega_1$ and $\lambda$, the probability that \eqref{eq:USF_error} holds may be very low. Conversely, in the recovery with the proposed operator $\MOh^D$, $\inner{\eta_{\bbf{k}},\psi_{\bbf{b},m}}$ satisfies $\inner{\eta_{\bbf{k}},\psi_{\bbf{b},m}}\sim \mathcal{N}\rb{0,\sigma^2\cdot \frac{2^N}{N^B}}$ \eqref{eq:noise3}, where the standard deviation can be made arbitrarily small by increasing the number of samples $N^B$ within each band $\bbf{b}$. In Section \ref{sect:simulations} this fact will be exploited to achieve significantly higher recovery performance for the proposed operator $\MOh^D$ compared to ideal modulo $\MO$.

\section{Numerical study}
\label{sect:simulations}

Let $\mathbf{V}=\sqb{\mathbf{v}_1,\mathbf{v}_2}$ be a randomly generated matrix such that $\norm{\mathbf{v}_1}_2=\norm{\mathbf{v}_1}_2=1$. The input $f:\R^2\rightarrow \R$ was restricted to two variables $x_1, x_2$ for visualisation purposes, and was generated as $f\rb{\mathbf{x}}=f_0\rb{\mathbf{Vx}}$, where 
\begin{equation}
    f_0\rb{x_1,x_2}=\sum_{\mathbf{k}\in\Z^2} c_{\mathbf{k}} \frac{\sin\rb{\Omega_1\rb{x_1-k_1\frac{\pi}{\Omega_1}}}}{{\Omega_1\rb{x_1-k_1\frac{\pi}{\Omega_1}}}}\cdot\frac{\sin\rb{\Omega_2\rb{x_2-k_2\frac{\pi}{\Omega_2}}}}{{\Omega_2\rb{x_2-k_2\frac{\pi}{\Omega_2}}}}.
\end{equation}
\begin{figure}
\begin{center}
\includegraphics[width=0.76\textwidth]{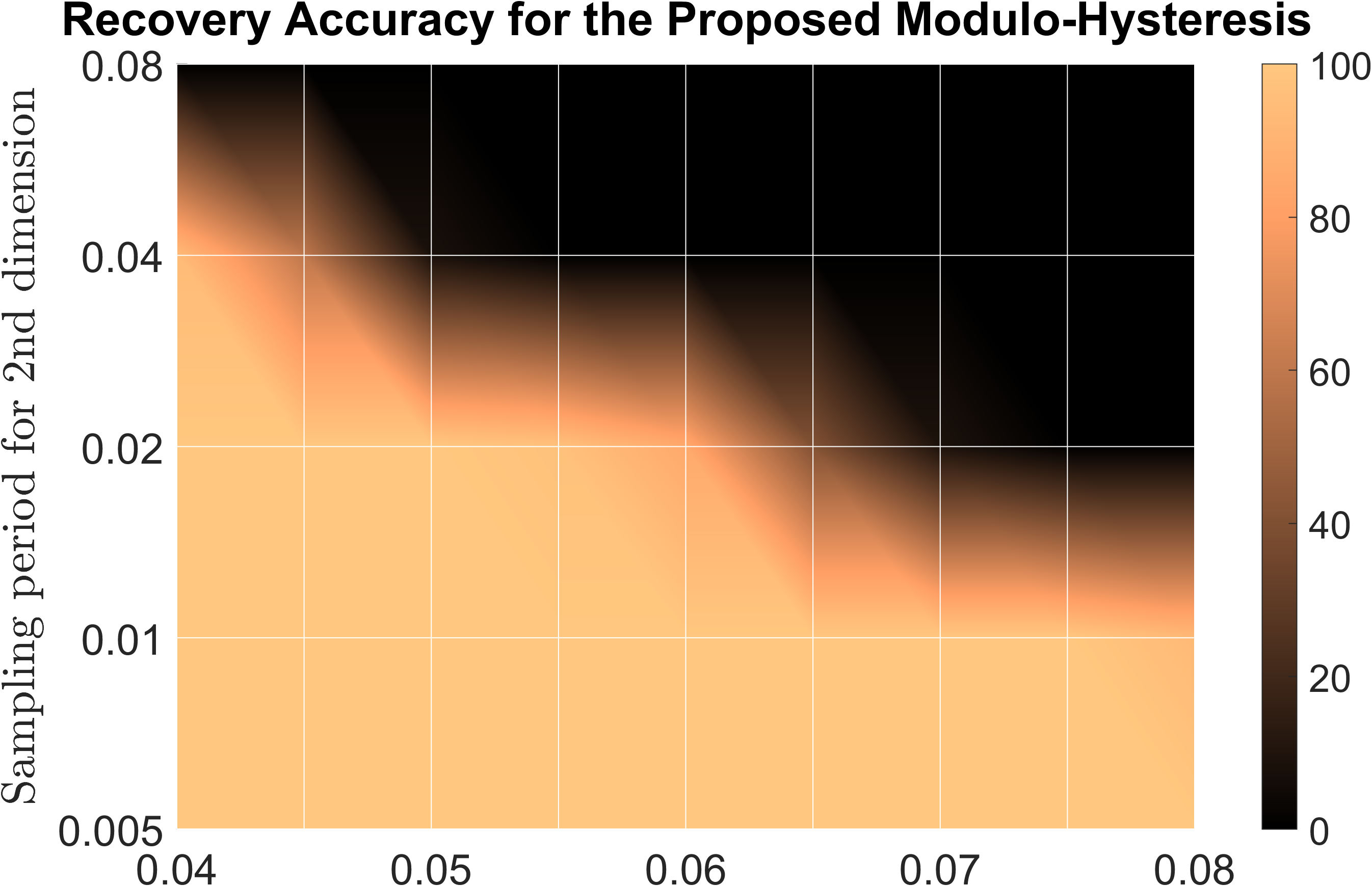}    
\vspace{0.7em}
\includegraphics[width=0.75\textwidth]{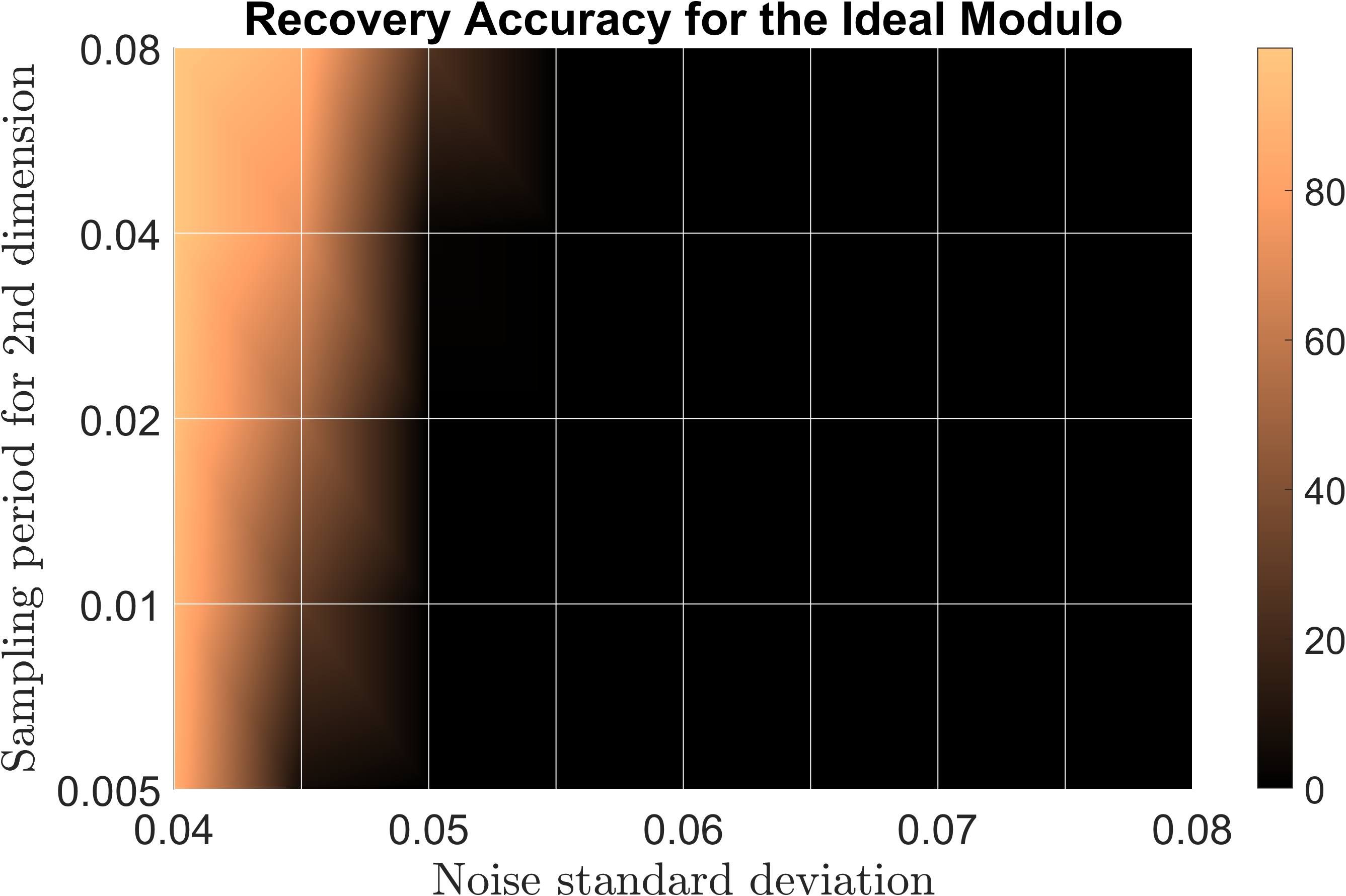}   
\end{center}
\caption{Reconstruction accuracy comparison for (a) the ideal modulo $\MO$ and (b) modulo-hysteresis $\MOh^D$ for a two-dimensional input.}
\label{fig:accuracy}
\end{figure}
\hspace{-0.35em}We selected $\Omega_1=\Omega_2=1$, and computed $f$ for $\mathbf{x}\in \sqb{-5,5}^2$. The coefficients $c_{\mathbf{k}}$ were randomly generated for $\vb{k_1}\leq1, \vb{k_2}\leq1$, drawn from the uniform distribution on $\sqb{-1,1}$. The dynamic range of $f$ is $\sqb{-1,1}$. We encoded $f$ using the ideal modulo $\MO$ with threshold $\lambda=0.3$ and the proposed modulo-hysteresis $\MOh^2$ with $\lambda=0.3, h=0.19, B=0.32$. The output samples, computed on lattice $\bs{\Lambda}$ with basis vectors $\mathbf{v}_1, \mathbf{v}_2$ and sampling periods $T_1,T_2$. We kept $T_1=0.02\ \mathrm{s}$ constant because its choice affects in a similar fashion recovery from $\MO$ and $\MOh^2$ (see \cite{Florescu:2022:J}). We varied $T_2$ in the range $\sqb{0.005\ s,0.08\ s}$. The output samples are
\begin{gather}
    y_1\sqb{\mathbf{k}}=\MO \rb{f\rb{\mathbf{VTk}}}+\eta\sqb{\mathbf{k}},\\
    y_2\sqb{\mathbf{k}}=\MOh^2 f\rb{\mathbf{VTk}}+\eta\sqb{\mathbf{k}},
\end{gather}
where $\eta\sqb{\mathbf{k}}\sim\mathcal{N}\rb{0,\sigma^2}$. 
We varied $\sigma$ in the range $\sqb{0.04,0.08}$, and recovered the input $\widetilde{\gamma}\sqb{\mathbf{k}}=f\rb{\mathbf{VTk}}+\eta\sqb{\mathbf{k}}$ up to a constant multiple of $2\lambda$ for $\MO$ and $2h$ for $\MOh^2$. We generated $100$ random inputs $f$ and $100$ noise sequences $\eta\sqb{\mathbf{k}}$ and counted the number of inputs correctly reconstructed using each method. In our context, correctly reconstructed means that the recovery conditions hold true. The results are depicted in \fig{fig:accuracy}. 
We note that, while the accuracy increases significantly for $\MOh^2$ for small $T_2$, as proven by Theorem \ref{th:noisy_input_rec}, for $\MO$ the reverse happens. This is because $\MO$ processes the input in a line-by-line fashion, and does not exploit in any way the higher resolution along dimension $x_2$. In fact, here a higher resolution simply adds more noise samples from sequence $\eta\sqb{\mathbf{k}}$, which are not filtered and thus increase the probability that \eqref{eq:USF_error} does not hold. We also note that, for larger sampling periods $T_2\sim0.08\ \mathrm{s}$,  $\MO$ performs slightly better for low noise, i.e.,  $\sigma<0.05$. This is explained by the fact that the modulo-hysteresis requirement $\vb{\inner{\eta_{\bbf{k}},\Delta^N\sqb{\cdot-m}}}<h/2-\rb{T_1\Omega_1 e}^N\norm{f}_\infty$ is more strict than \eqref{eq:USF_error} given that $h/2<\lambda$. This is a small trade-off that enables $\MOh^2$ to handle arbitrarily large values of $\sigma$ for small enough sampling periods $T_2$.

\section{Proofs}
\label{sect:proofs}

\subsection{Multi-Dimensional Modulo Properties}
\label{sect:proofs_prop}

\if\BdlSlicesInAppendix\FL
\else
    \begin{proof}[\textbf{Proof for Proposition \ref{prop:pw_for_slices} (Bandlimited slices)}]
    \input{Proofs/BdlSlices}
    \end{proof}
\fi

\if\FoldSepInAppendix\FL
\else
    \begin{proof}[\textbf{Proof for Proposition \ref{prop:fold_sep_multid} (Folding time separation)}]
    \input{Proofs/FoldSep}
    \end{proof}
\fi

\if\WellDefInAppendix\FL
\else
    \begin{proof}[\textbf{Proof for Proposition \ref{prop:well_defined_op} (Well-defined operator)}]
    \input{Proofs/WellDef}
    \end{proof}
\fi

\if\ModDynRangeInAppendix\FL
\else
    \begin{proof}[\textbf{Proof for Proposition \ref{prop:mod_dyn_range} (Modulo output dynamic range)}]
    \input{Proofs/ModDynRange}
    \end{proof}
\fi

\if\BoundDelfBInAppendix\FL
\else
    \begin{proof}[\textbf{Proof for Corollary \ref{cor:bound_fB} (Bound for intra-band variation)}]
    \input{Proofs/BoundDelfB}
    \end{proof}
\fi

\subsection{Multi-Dimensional Modulo Recovery}
\label{sect:proofs_rec}
\if\FoldDetInAppendix\FL
\else
    \begin{proof}[\textbf{Proof for Proposition \ref{prop:fold_det} (Detection of folding times)}]
    \input{Proofs/FoldDet}
    \end{proof}
\fi

\if\MbDetInAppendix\FL
\else
    \begin{proof}[\textbf{Proof for Proposition \ref{prop:Mb_det} (Detection of constants $M_{\bbf{b}}$)}]
    \input{Proofs/MbDet}
    \end{proof}
\fi

\if\InputRecCleanInAppendix\FL
\else
    \begin{proof}[\textbf{Proof for Theorem \ref{th:input_rec} (Noiseless input recovery)}]
    \input{Proofs/InputRecClean}
    \end{proof}
\fi

\section{Conclusion}

The Unlimited Sampling Framework (USF) provides sampling rate guarantees that allow tackling high dynamic range signals in the \D case. For \MD signals, USF is typically applied sequentially, thus not exploiting the \MD structure of the input. 
In this paper, we 
\begin{itemize}
    \item introduced the first \MD modulo operator and associated input reconstruction method from lattice samples,
    \item derived sampling rate conditions under which the reconstruction is perfect in the noiseless scenario,
    \item provided probability error bounds under Gaussian noise assumption,
    \item showed numerically that, while USF does not allow noise amplitudes larger than the modulo threshold, the proposed approach allows arbitrarily high noise for sufficiently small sampling times.
\end{itemize}
This work can be extended in a number of ways
\begin{enumerate}
\item It can be coupled with modulo denoising approaches such as \cite{Tyagi:2022:J} to yield enhanced reconstruction algorithms.
\item While it is assumed that the input is bandlimited, this work can be extended for inputs generated with B-splines or sparse inputs.
\item The model in our work can be extended to a wider range of models such as modulo neuromorphic architectures, that could exploit \MD inputs in a similar way to the biological systems \cite{Gallego:2020} \cite{Florescu:2021:C}.
\item Alternative sampling mechanisms that would benefit from a \MD modulo operation include one-bit sampling \cite{Graf:2019:C} or average sampling \cite{Florescu:2022:C}.
\item The current line of work can lead to a the implementation of a new \MD hardware prototype.
\end{enumerate}

\end{document}

%% file: 0packages.tex
\usepackage[english]{babel}

\usepackage{algorithmic}
\usepackage[ruled,vlined]{algorithm2e}
\usepackage{amsfonts}
\usepackage{amsmath}
\usepackage{amsthm}
\usepackage{amssymb}
\usepackage{bbm}
\usepackage{enumitem}
\usepackage{epstopdf}
\usepackage{float}
\usepackage[T1]{fontenc}
\usepackage{glossaries}
\usepackage{graphicx}
\usepackage{times}
\usepackage{latexsym}
\usepackage{mathtools}
\usepackage{multirow}
\usepackage{makeidx}
\usepackage{mathrsfs} 
\usepackage{nomencl}
\usepackage{rotating}
\usepackage{textcomp}
\usepackage{varioref}
\usepackage{soul}
\usepackage{xspace}
\usepackage[usenames,dvipsnames,table]{xcolor}

\DeclareGraphicsExtensions{.eps,.png}

%% file: 0shortcuts.tex
\renewcommand{\thefootnote}{\normalsize \arabic{footnote}} 	
\newtheorem{theo}{Theorem}
\newtheorem{conj}{Conjecture}
\newtheorem{definition}{Definition}
\newtheorem{prop}{Proposition}
\newtheorem{lem}{Lemma}
\newtheorem{cor}{Corollary}
\newtheorem{remark}{Remark}
\newtheorem{example}{Example}

\newcommand{\pw}{\mathsf{PW}_{\Omega}\rb{\R}}
\newcommand{\pwd}{\mathsf{PW}_{\bs{\Omega}}\rb{\RD}}

\newcommand{\kz}{k\in\mathbb{Z}}
\newcommand{\akz}{\forall k \in \mathbb{Z}}
\newcommand{\tu}{\mathcal{T}_u}
\newcommand{\db}{\bar{\delta}}
\newcommand\bb[1]{{\color{blue}#1}}
\newcommand{\Z}{\mathbb{Z}}
\newcommand{\N}{\mathbb{N}}
\newcommand{\R}{\mathbb{R}}
\newcommand{\LR}{L^2(\mathbb{R})}
\newcommand{\LO}{L^2([-\Omega,\Omega])}
\newcommand{\re}{\mathbb{R}}
\newcommand{\co}{\mathbb{C}}
\newcommand{\cH}{\mathcal{H}}
\newcommand{\Tu}{\mathcal{T}_u}
\newcommand{\F}{\mathcal{F}}
\newcommand{\bs}{\boldsymbol} 

\newcommand{\eq}{\triangleq}
\newcommand{\s}{ISI}
\newcommand{\xaux}{\bbs{x}^*\hspace{-0.2em}}
\newcommand{\chiaux}{\bbs{\chi}^*\hspace{-0.2em}}
\newcommand{\uchiaux}{\bbs{\uchi}^*\hspace{-0.2em}}

\newcommand{\MO}[0]{\mathscr{M}_\lambda}
\newcommand{\fig}[1]{Fig.~\ref{#1}}

\newcommand{\Sz}{\mathbb{S}_\varepsilon}
\newcommand{\tilSz}{\widetilde{\mathbb{S}}_\varepsilon}
\newcommand{\kernel}{\mathcal{B}}

\newcommand{\setsep}{\vert}

\newcommand{\cI}{\mathcal{I}}
\newcommand{\cR}{\mathcal{R}}
\newcommand{\cP}{\mathcal{P}}
\newcommand{\cB}{\mathcal{B}}

\newcommand{\RD}{\mathbb{R}^D}
\newcommand{\RDred}{\mathbb{R}^{D-1}}

\newcommand{\ZD}{\mathbb{Z}^D}
\newcommand{\ZDred}{\mathbb{Z}^{D-1}}

\newcommand{\D}{one-dimensional }
\newcommand{\Dcaps}{One-Dimensional }

\newcommand{\MD}{multi-dimensional\xspace}
\newcommand{\MDcaps}{Multi-Dimensional\xspace}

\newcommand{\MOh}[0]{\mathscr{M}_{\boldsymbol{\mathsf{H}}}}
\newcommand{\nbr}{\mathrm{Neighbors}\hspace{-0.2em}\rb{\bar{\mathbf{b}}}}

\newcommand{\df}{\mathcal{D} f}
\newcommand{\dfb}{\mathcal{D}_{\bbf{b}} f}

\newcommand{\fe}[1]{\left[\kern-0.30em\left[#1  \right]\kern-0.30em\right]}

\newcommand{\vb}[1]{\left\lvert #1 \right\rvert}
\newcommand{\rb}[1]{\left( #1 \right)}
\newcommand{\sqb}[1]{\left[ #1 \right]}
\newcommand{\cb}[1]{\left\lbrace #1 \right\rbrace}
\newcommand{\floor}[1]{\left\lfloor #1 \right\rfloor}
\newcommand{\ceil}[1]{\left\lceil #1 \right\rceil}
\newcommand{\inner}[1]{\left\langle #1 \right\rangle}
\newcommand{\bbs}[1]{\bar{\bs{ #1 }}}
\newcommand{\bbf}[1]{\bar{\mathbf{ #1 }}}

\newcommand{\xinR}{\bbs{x}\in\cR_{\bbs{b}}}
\newcommand{\uchi}{\raisebox{1.5pt}{$\chi$}}
\renewcommand\tilde{\widetilde}

\newcommand{\ab}[1]{{\color{red} #1}}
\newcommand{\dor}[1]{{\color{ForestGreen} #1}}

\DeclarePairedDelimiter{\norm}{\Vert}{\Vert}
\DeclarePairedDelimiter{\abs}{\left|}{\right|}
\DeclarePairedDelimiter{\Prod}{\langle}{\rangle}

\newcommand{\PW}[1]{\mathsf{PW}_{#1}}
\newcommand{\iPW}[2]{#1 \in \mathsf{PW}_{#2}}

\newcommand{\pdf}{p.d.f.\ }

\def\ind{\mathbbmtt{1}}

\def\figmode    {1}         
\def\PH         {0}         
\def\stabilization	{0}
\def\True		{1}

\def\BdlSlicesInAppendix {1}
\def\FoldSepInAppendix {1}
\def\WellDefInAppendix {1}
\def\ModDynRangeInAppendix {1}
\def\BoundDelfBInAppendix {1}
\def\FoldDetInAppendix {1}
\def\MbDetInAppendix {1}
\def\InputRecCleanInAppendix {1}
\def\TR                 {1}
\def\FL                 {0}

\renewcommand\hat\widehat
\renewcommand\geq\geqslant
\renewcommand\leq\leqslant

%% file: 0Notation.tex
\paragraph{Notation}

For $x\in\R$, $[\![ x ]\!]=x- \lfloor x \rfloor$ denotes the fractional part of $x$ and $\floor{x}$ is the floor function. For a set $\mathbb{S}$, $\ind_{\mathbb{S}}$ is the indicator function and $\mathrm{cl}\rb{\mathbb{S}}$ is the set closure. The set of real and integer numbers are $\R$ and $\Z$, respectively. Let $\R^*=\R\setminus\cb{0}$ and $\Z^*=\Z\setminus\cb{0}$, and let the sets restricted to positive numbers $x\geq0$ be $\R_+$ and $\Z_+$. We denote by $\emptyset$ the empty set.

We use bold lowercase for vectors such as $\mathbf{x}=\sqb{x_1,\dots,x_D}^\top$, assumed to be column vectors unless otherwise specified. Matrices are denoted by bold uppercase, e.g., $\mathbf{M}\in\R^{D\times D}$. The element on line $k_1$ and column $k_2$ in matrix $\mathbf{M}$ is denoted by $\sqb{\mathbf{M}}_{k_1,k_2}$. Unless specified otherwise, we use notation $\mathbf{x}$ to denote a vector $\mathbf{x}\in\ZD$ and  $\bbf{x}$ to denote a vector $\bbf{x}\in\ZDred$. When used in the same context, $\bbf{x}$ denotes the last $D-1$ coordinates of $\mathbf{x}$ such that $\mathbf{x}=\sqb{x_1,\bbf{x}}$. Similarly, we denote by $\bbf{V}\in \R^{D\times \rb{D-1}}$ a matrix containing the last $D-1$ columns of matrix $\mathbf{V}\in\R^{D\times D}$. For two vectors $\mathbf{v}_1,\mathbf{v}_2\in\RD$ we denote their inner product by $\inner{\mathbf{v}_1,\mathbf{v}_2} =\sum_{d=1}^D v_{1,d}\cdot\mathtt{conj}({v}_{2,d})=\mathbf{v}_1^\top \cdot \mathtt{conj}(\mathbf{v}_1)$, where $\mathtt{conj}$ is the complex-conjugate. Norm $\norm{\mathbf{v}}_2$ denotes the Euclidean norm for a vector $\mathbf{v}\in\RD$ and norm $\norm{\mathbf{v}}_\infty$ is defined as $\norm{\mathbf{v}}_\infty=\max_{d=1,\dots,D}\vb{\sqb{\mathbf{v}}_d}$. We denote by $\det\rb{\mathbf{V}}$ the determinant of matrix $\mathbf{V}$. We denote by $\mathbf{T}=\mathrm{diag}\cb{T_1,\dots,T_D}$ a matrix with $T_1,\dots,T_D$ on the main diagonal and $0$ otherwise.

For a function $f:\RD\rightarrow\R$, $\norm{f}_2$ and $\norm{f}_\infty$ represent the $L^2\rb{\R^D}$ and $L^\infty\rb{\R^D}$ norms, respectively. We denote by $\mathcal{F} f$ the Fourier transform applied to $f$, defined as
\begin{equation}
    \mathcal{F} f\rb{\boldsymbol{\omega}} = \int_{\mathbb{R}^D} f\rb{\mathbf{x}} e^{-\jmath \inner{\bs{\omega}, \mathbf{x}}} d\mathbf{x},
\end{equation}    
where $\mathbf{x}=\sqb{x_1,\cdots,x_D}^\top$ and $\boldsymbol{\omega}=\sqb{\omega_1,\cdots,\omega_D}^\top$. The inverse Fourier transform $\mathcal{F}^{-1} F$ is defined as
\begin{equation}
    \mathcal{F}^{-1} F\rb{\mathbf{x}} = \frac{1}{\rb{2\pi}^D}\int_{\mathbb{R}^D} F\rb{\boldsymbol{\omega}} e^{\jmath \inner{\bs{\omega}, \mathbf{x}}} d\boldsymbol{\omega}.
\end{equation}    
The support of sequence $\psi$ is denoted by $\mathrm{supp}\rb{\psi}$ and the support of a function $f\rb{\mathbf{x}}$ is $\mathrm{supp}\rb{f}$. For two \MD sequences $\gamma_1, \gamma_2:\ZD\rightarrow \R$, $\inner{\gamma_1,\gamma_2}$ denotes the \MD inner product defined as $\inner{\gamma_1,\gamma_2}=\sum_{\mathbf{k}\in\ZD} \gamma_1\sqb{\mathbf{k}} \mathtt{conj}(\gamma_2\sqb{\mathbf{k}})$. The coefficients for the forward finite difference of order $N$ are denoted by
$\Delta^N\sqb{k}$. Specifically, it is defined as $\Delta^N\sqb{k}=\Delta_-^N\sqb{-k}$, where $\Delta^N_-$ is defined recursively as $\Delta^{N+1}_-\sqb{k}=\Delta^N_- \ast \Delta^1_-\sqb{k}$, where $\Delta^1_-\sqb{-1}=1, \Delta^1_-\sqb{0}=-1$, $\Delta^1_-\sqb{k}=0, k\in\Z\setminus\cb{-1,1}$.

We denote by $\mathbf{v}_1,\dots,\mathbf{v}_D$ the set of vectors defining a lattice $\bs{\Lambda}=\cb{\mathbf{V T k},\mathbf{k}\in\ZD}$ where $\mathbf{V}=\sqb{\mathbf{v}_1,\dots,\mathbf{v}_D}, \mathbf{v}_d\in\R^D, \forall d \in\cb{1,\dots,D}$ and $\mathbf{T}=\mathrm{diag}\cb{T_1,\dots,T_D}$, where $T_d$ denotes the sampling period across dimension $d\in\cb{1,\dots,D}$. Without loss of generality, we assume that $\mathbf{v}_d$ are \emph{versors}, i.e., $\norm{\mathbf{v}_d}_{2}=1,\forall d \in\cb{1,\dots,D}$. The vectors are assumed linearly independent, and thus $\mathbf{V}$ is invertible. Therefore, a function $f:\RD\rightarrow \R$ can be equivalently evaluated using Cartesian coordinates as $f\rb{\mathbf{x}_c}$ or  lattice coordinates as $f\rb{\mathbf{Vx}_v}$ such that $\mathbf{x}_v=\mathbf{V}^{-1} \mathbf{x}_c$. Unless specified otherwise, $\gamma\sqb{\mathbf{k}}, \mathbf{k}\in\ZD$, denote the samples of the input function $f$ on lattice $\bs{\Lambda}$ such that $\gamma\sqb{\mathbf{k}}=f\rb{\mathbf{VTk}}$. 
The dual lattice $\hat{\bs{\Lambda}}$ is defined as $\hat{\bs{\Lambda}}=\cb{\hat{\mathbf{V}} \mathbf{T}^{-1} \mathbf{k}\setsep \mathbf{k}\in\Z^D}$, where  $\hat{\mathbf{V}}=\mathbf{V}^{-\top}=\rb{\mathbf{V}^{-1}}^\top$ and $\mathbf{T}^{-1}=\mathrm{diag}\cb{1/T_1,\dots,1/T_D}$. 

The $D$-dimensional Paley-Wiener space $\pwd$ of bandwidth $\bs{\Omega}$ relative to lattice $\bs{\Lambda}$ consists of functions $f\in L^2\rb{\mathbb{R}^D}$ such that
\begin{equation}
\label{eq:pwd}
    \mathrm{supp} \rb{\mathcal{F} f} \subseteq \cb{{\hat{\mathbf{V}}}\bs{\omega}=\sum_{d=1}^D \omega_d\hat{v}_d\in\mathbb{R}^D \setsep \vb{\omega_d}<\Omega_d,\forall d \in \cb{1,\dots,D}}.
\end{equation}
For the one-dimensional case $D=1$, the lattice matrices $\hat{\mathbf{V}}$ reduce to the trivial case $\hat{\mathbf{V}}=\mathbf{V}=\sqb{1}$ and $\pwd$ reduces to the classical \D Paley-Wiener space $\pw$.

For a random variable $\eta$ we denote by \pdf its probability density function. We denote by $\mathcal{N}\rb{\mu,\sigma^2}$ the normal distribution with mean $\mu$ and standard deviation $\sigma$. A random variable $\eta$ drawn from the Gaussian normal distribution is denoted as $\eta\sim\mathcal{N}\rb{\mu,\sigma^2}$.

%% file: Proofs/BdlSlices.tex
We consider the slice $f_{\mathbf{V}}= f\rb{\mathbf{Vx}}$ along dimension $x_d$, with $d\in\cb{1,\dots,D}$ fixed. To this end, we apply the $(D-1)$ -- dimensional inverse Fourier transform to $f_{\mathbf{V}}$ corresponding to all variables apart from $x_d$, such that
\begin{align}
\label{eq:FfV}
    \frac{1}{\rb{2\pi}^{D-1}}&\int_{\R^{D-1}} \mathcal{F} f_{\mathbf{V}}\rb{\mathbf{\omega}} e^{\jmath\inner{\mathbf{\omega},\mathbf{x}}} d\omega_1 \dots d\omega_{d-1} d\omega_{d+1} \dots d\omega_D\\
    =&\int_{\R} f_{\mathbf{V}}\rb{\mathbf{x}} e^{-\jmath \omega_d x_d} dx_d\\
    =&\int_{\R} f\rb{\sum\nolimits_{d=1}^D x_d \mathbf{v}_d}e^{-\jmath \omega_d x_d}d x_d=\mathcal{F} g_d\rb{\omega_d}.
\end{align}
Therefore, the spectrum of $g_d$ depends on the spectrum of $f_{\mathbf{V}}\rb{\mathbf{x}}$, which will be evaluated in the following. To this end, via the change of variable $\mathbf{x}=\mathbf{Vx}^*$, we get that
\begin{align}
\label{eq:fourier_fv}
\begin{split}
    \mathcal{F}f\rb{\mathbf{\omega}}=\int_{\R^D} f\rb{\mathbf{x}}e^{-\jmath\inner{\mathbf{\omega},\mathbf{x}}}d\mathbf{x}
    &=\int_{\R^D} f\rb{\mathbf{Vx}^*}e^{-\jmath\inner{\mathbf{\omega},\mathbf{Vx}^*}}\vb{\det\rb{\mathbf{V}}}d\mathbf{x}^*\\
    &=\int_{\R^D} f\rb{\mathbf{Vx}^*}e^{-\jmath\inner{\mathbf{\mathbf{V}^\top\omega},\mathbf{x}^*}}\vb{\det\rb{\mathbf{V}}}d\mathbf{x}^*\\
    &=\vb{\det\rb{\mathbf{V}}}\cdot\mathcal{F} f_{\mathbf{V}}\rb{\mathbf{V}^\top\mathbf{\omega}}.
\end{split}
\end{align}
Using $f\in\pwd$ we get that $\mathrm{supp}\rb{\mathcal{F}f}\subseteq \mathbf{V}^{-\top}\cdot\prod_{d=1}^D \rb{-\Omega_d,\Omega_d}$. Then, via \eqref{eq:fourier_fv}, we have that $\mathrm{supp} \rb{\mathcal{F}f_{\mathbf{V}}}\subseteq \prod_{d=1}^D \rb{-\Omega_d,\Omega_d}$, and therefore $\mathcal{F} f_{\mathbf{V}}\rb{\mathbf{\omega}}=0,\forall \omega_d\in\R, \vb{\omega_d}>\Omega_d$ in \eqref{eq:FfV}. It follows that $\mathcal{F} g_d\rb{\omega_d}=0,\forall \omega_d\in\R, \vb{\omega_d}>\Omega_d$, $d\neq 1$. Furthermore, choosing $d=1$ gives us $\mathcal{F} f_{\bbf{x}}\rb{\omega_1}=0,\forall \omega_1\in\R, \vb{\omega_1}>\Omega_1$, which finalizes the proof.

%% file: Proofs/FoldSep.tex
We first show an intermediate result in the following lemma.
\begin{lem}
\label{lem:gr_Lipschitz_cont}

For a fixed vector $\bbf{b}\in \ZDred$ indicating the modulo band and $r\in\cb{1,\dots,R}$, let  $g_{\bbf{x},r}:\R\rightarrow \R$ and $g_r:\R\rightarrow\R$ be two functions defined as
\begin{align}
    g_{\bbf{x},r}\rb{x_1}&\triangleq f_{\bbf{x}}\rb{x_1}-\varepsilon_{\bbf{b},r-1}\rb{x_1},\quad \forall \mathbf{x}\in\R^D,
    \label{eq:gxr}\\
    g_r\rb{x_1}&\triangleq \sup_{\bbf{x}\in\cR_{\bbf{b}}}\vb{g_{\bbf{x},r}\rb{x_1}}, \quad \forall x_1\in\R.
    \label{eq:gr}
\end{align}
Then $g_{\bbf{x},r}$ and $\vb{g_{\bbf{x},r}}$ are Lipschitz-continuous as functions of $\mathbf{x}=\sqb{x_1,\bbf{x}}\in\ZD$ and $g_r$ is Lipschitz-continuous as function of $x_1\in \R$. Furthermore, the Lipschitz constant in all cases is $\norm{\boldsymbol{\Omega}}_2 \cdot \norm{f}_{\infty}$.
\end{lem}
\begin{proof}
We begin by deriving a bound for the Lischitz constant of function $f_{\mathbf{V}}=f\rb{\mathbf{V}\cdot}$. Given that $f$ is differentiable, according to the \emph{mean value theorem} 
\begin{equation}
    f_{\mathbf{V}}\rb{\mathbf{x}}-f_{\mathbf{V}}\rb{\boldsymbol{\uchi}}=\inner{\nabla f_{\mathbf{V}}\rb{\mathbf{c}},\mathbf{x}-\boldsymbol{ \uchi}},\quad \forall \mathbf{x},\boldsymbol{\uchi}\in \RD,
\end{equation}
where  $\mathbf{c}=\alpha \mathbf{x}+\rb{1-\alpha}\boldsymbol{\uchi}, \alpha\in\sqb{0,1}$ is an intermediate point on the segment joining $\mathbf{x}$ and $\boldsymbol{\uchi}$. Due to the \emph{Cauchy-Schwartz} inequality
\begin{equation}
    \vb{f_{\mathbf{V}}\rb{\mathbf{x}}-f_{\mathbf{V}}\rb{\boldsymbol{\uchi}}}\leq \norm{\nabla f_{\mathbf{V}}\rb{\mathbf{ c}}}_2\cdot\norm{\mathbf{x}-\boldsymbol{\uchi}}_2.
\end{equation}
Next, we use the Bern\v{s}te\u{\i}n bounds in \eqref{eq:Bernstein_multiD} to derive a bound for $\norm{\nabla f_{\mathbf{V}}}$ as follows
\begin{align}
\norm{\nabla f_{\mathbf{V}}}_2^2=\sum_{d=1}^D \vb{\frac{\partial}{\partial x_d}f_{\mathbf{V}}\rb{\mathbf{c}}}^2\leq \sum_{d=1}^D \Omega_d^2\norm{f}^2_\infty
\end{align}
which means $f_{\mathbf{V}}$ is a Lipschitz-continuous function satisfying
\begin{equation}
    \vb{f_{\mathbf{V}}\rb{\mathbf{x}}-f_{\mathbf{V}}\rb{\boldsymbol{\uchi}}}\leq \norm{\mathbf{\Omega}}_2\cdot \norm{f}_\infty\cdot\norm{\mathbf{x}-\boldsymbol{\uchi}}_2,\quad \forall \mathbf{x},\boldsymbol{\uchi}\in\RD.
    \label{eq:Lipschitz_bound}
\end{equation}

For all $\bbf{x},\bbs{\uchi}\in\RDred$ and $\forall x_1,\uchi_1>\tau_{\bbf{b},r-1}$ we have $\varepsilon_{\bbf{b},r-1}\rb{x_1}=\varepsilon_{\bbf{b},r-1}\rb{\uchi_1}$, which implies
\begin{align}
\begin{split}
    \vb{\vb{g_{\bbf{x},r}\rb{x_1}}-\lvert{g_{\bbs{\uchi},r}\rb{\uchi_1}}\rvert}
    &\leq\vb{g_{\bbf{x},r}\rb{x_1}-g_{\bbs{\uchi},r}\rb{\uchi_1}}
    =\vb{f_{\bbf{x}}\rb{x_1}-f_{\bbs{\uchi}}\rb{\uchi_1}}\\
    &=\vb{f_{\mathbf{V}}\rb{\mathbf{x}}-f_{\mathbf{V}}\rb{\boldsymbol{\uchi}}}
    \leq \norm{\mathbf{\Omega}}_2\cdot \norm{f}_\infty\cdot\norm{\mathbf{x}-\boldsymbol{\uchi}}_2.
\end{split}
\label{eq:Lipsch_absg}
\end{align}
The final step is to show that $g_r\rb{x_1}$ is also Lipschitz, i.e., that the supremum does not change the Lipschitz constant. To this end, we note the following two properties of the supremum. Specifically, for $\forall \epsilon>0$ and $\forall x_1,\uchi_1\in\R$, $\exists \bbf{x}_\epsilon,\bbs{\uchi}_\epsilon \in \cR_{\bbf{b}}$ such that
\begin{gather}
\begin{split}
    \vb{g_{\bbs{\uchi}_\epsilon}\hspace{-0.3em}\rb{x_1}}\leq g_r\rb{x_1}<\vb{g_{\bbf{x}_\epsilon}\hspace{-0.3em}\rb{x_1}}+\epsilon,\\
    \vb{g_{\bbf{x}_\epsilon}\hspace{-0.3em}\rb{\uchi_1}}\leq g_r\rb{\uchi_1}<\vb{g_{\bbs{\uchi}_\epsilon}\hspace{-0.3em}\rb{\uchi_1}}+\epsilon.
\end{split}
\label{eq:supp_conds}
\end{gather}
We derive that $ g_r\rb{x_1}-g_r\rb{\uchi_1}\leq \vb{g_{\bbf{x}_\epsilon}\rb{x_1}}+\epsilon - \vb{g_{\bbf{x}_\epsilon}\rb{\uchi_1}}$. Similarly, we get $ g_r\rb{\uchi_1}-g_r\rb{x_1}\leq \vb{g_{\bbs{\uchi}_\epsilon}\rb{\uchi_1}}+\epsilon - \vb{g_{\bbs{\uchi}_\epsilon}\rb{x_1}}$. Finally, we restrict $x_1,\uchi_1$ in \eqref{eq:supp_conds} to satisfy $x_1,\uchi_1 > \tau_{\bbf{b},r-1}$ and select $\bbf{x},\bbs{\uchi}$ in \eqref{eq:Lipsch_absg} as $\bbf{x}=\bbf{x}_\epsilon,\bbs{\uchi}=\bbs{\uchi}_\epsilon$, which yields
\begin{align*}
   \vb{g_r\rb{x_1}-g_r\rb{\uchi_1}}
   &\leq \max\cb{\vb{g_{\bbf{x}_\epsilon}\rb{x_1}} - \vb{g_{\bbf{x}_\epsilon}\rb{\uchi_1}},\vb{g_{\bbs{\uchi}_\epsilon}\rb{\uchi_1}} - \vb{g_{\bbs{\uchi}_\epsilon}\rb{x_1}}}+\epsilon\\
   &\leq\max\cb{\Big\lvert{\vb{g_{\bbf{x}_\epsilon}\rb{x_1}} - \vb{g_{\bbf{x}_\epsilon}\rb{\uchi_1}}}\Big\rvert,\vb{\vb{g_{\bbs{\uchi}_\epsilon}\rb{\uchi_1}} - \vb{g_{\bbs{\uchi}_\epsilon}\rb{x_1}}}}+\epsilon\\
   &\leq \vb{{\Omega}_1}\cdot \norm{f}_\infty\cdot\vb{x_1-\uchi_1}+\epsilon,\quad\forall \epsilon>0.
\end{align*}
Taking $\epsilon\rightarrow0$ above proves the required result.
\end{proof}

We begin by evaluating $g_{\bbf{x},1}\rb{\tau_{\bbf{b},0}}$ 
\begin{equation}
g_{\bbf{x},1}\rb{\tau_{\bbf{b},0}}=f_{\bbf{x}}\rb{\tau_{\bbf{b},0}}-hM_{\bbf{b}}=f_{\bbf{x}}\rb{\tau_{\bbf{b},0}}-h\floor{\frac{\inf_{\bbs{\uchi}\in\cR_{\bbf{b}}}f_{\bbs{\uchi}}\rb{\tau_{\bbf{b},0}}+\lambda}{h}}+h.    
\end{equation}
Using that $x-1<\floor{x}\leq x$ we derive
\begin{align}
\begin{split}
    \inf_{\bbf{x}\in\cR_{\bbf{b}}} g_{\bbf{x},1}\rb{\tau_{\bbf{b},0}}&\geq \inf_{\bbf{x}\in\cR_{\bbf{b}}} \sqb{f_{\bbf{x}}\rb{\tau_{\bbf{b},0}}-h\tfrac{\inf_{\bbs{\uchi}\in\cR_{\bbf{b}}}f_{\bbs{\uchi}}\rb{\tau_{\bbf{b},0}}+\lambda}{h}+h}\\
    &= -\lambda+h\\
    \sup_{\bbf{x}\in\cR_{\bbf{b}}} g_{\bbf{x},1}\rb{\tau_{\bbf{b},0}}&\leq \sup_{\bbf{x}\in\cR_{\bbf{b}}} \sqb{f_{\bbf{x}}\rb{\tau_{\bbf{b},0}}-h\tfrac{\inf_{\bbs{\uchi}\in\cR_{\bbf{b}}}f_{\bbs{\uchi}}\rb{\tau_{\bbf{b},0}}+\lambda}{h}+2h}\\
    &\leq-\lambda+2h+\df \leq \lambda-h.
\end{split}
\label{eq:init}
\end{align}
Then, given that
\[g_1\rb{\tau_{\bbf{b},0}}=\sup_{\bbf{x}\in\cR_{\bbf{b}}}\vb{g_{\bbf{x},1}\rb{\tau_{\bbf{b},0}}}=\max\cb{\sup_{\bbf{x}\in\cR_{\bbf{b}}}{g_{\bbf{x},1}\rb{\tau_{\bbf{b},0}}},-\inf_{\bbf{x}\in\cR_{\bbf{b}}}{g_{\bbf{x},1}\rb{\tau_{\bbf{b},0}}}},
\]
we get $g_1\rb{\tau_{\bbf{b},0}}\leq \lambda-h$. Furthermore, $g_1$ is Lipschitz-continuous due to Lemma \ref{lem:gr_Lipschitz_cont}, and therefore continuous.  Thus, given that $\tau_{\bbf{b},1}=\inf \cb{x_1>\tau_{\bbf{b},0}\setsep g_1\rb{x_1}=\lambda}$ is true by definition, we get $g_1\rb{\tau_{\bbf{b},1}}=\lambda$. 

We then investigate the variation of function $g_1$ between folding times $\tau_{\bbf{b},0}$ and $\tau_{\bbf{b},1}$ as follows. Using Lemma \ref{lem:gr_Lipschitz_cont}, we get \[g_1\rb{\tau_{\bbf{b},0}}-g_1\rb{\tau_{\bbf{b},1}}\leq \Omega_1 \norm{f}_{\infty} \cdot \vb{\tau_{\bbf{b},1}-\tau_{\bbf{b},0}},\] 
and then we use that $g_1(\tau_{\bbf{b},1})=\lambda$ and $g_1(\tau_{\bbf{b},0})<\lambda-h$ to derive
\[h\leq\vb{g_1\rb{\tau_{\bbf{b},0}}-g_1\rb{\tau_{\bbf{b},1}}}.\]
Therefore the folding times satisfy the following bound
\begin{equation}
\vb{\tau_{\bbf{b},1}-\tau_{\bbf{b},0}}\geq \frac{h }{\Omega_1 \norm{f}_\infty}.
\label{eq:separation_MD_aux}
\end{equation}
We note the resemblance between \eqref{eq:separation_MD_aux} the \D case \eqref{eq:separation}. Next, we will show by induction that $g_r\rb{\tau_{\bbf{b},r-1}}\leq \lambda-h,\  g_r\rb{\tau_{\bbf{b},r}}= \lambda$ and then $\vb{\tau_{\bbf{b},r}-\tau_{\bbf{b},r-1}}\geq \frac{h}{\Omega_1 \norm{f}_\infty}$ holds where sequence $\cb{\tau_{\bbf{b},1},\dots,\tau_{\bbf{b},r}}$ is computed according to \eqref{eq:tau_recurrent}. The \emph{base case} $r=1$ is shown in the derivation to \eqref{eq:separation_MD_aux}. For the \emph{induction step} we proceed with computing $g_{r+1}\rb{\tau_{\bbf{b},r}}$.
\begin{gather}
    g_{r+1}\rb{\tau_{\bbf{b},r}}=\sup_{\xinR} \vb{g_{\bbf{x},r+1}\rb{\tau_{\bbf{b},r}}}=\sup_{\xinR} \vb{f_{\bbf{x}}\rb{\tau_{\bbf{b},r}}-\varepsilon_{\bbf{b},r}\rb{\tau_{\bbf{b},r}}}\\
    =\sup_{\xinR} \vb{f_{\bbf{x}}\rb{\tau_{\bbf{b},r}}-\varepsilon_{\bbf{b},r-1}\rb{\tau_{\bbf{b},r}}+h s_{\bbf{b},r} \ind_{\left[\tau_{\bbf{b},r}\infty\right)}\rb{\tau_{\bbf{b},r}}}\\
    =\sup_{\xinR}\vb{g_{\bbf{x},r}\rb{\tau_{\bbf{b},r}}-h s_{\bbf{b},r}},\ \mathrm{where}\ s_{\bbf{b},r}=\mathrm{sign} \sqb{g_{\bbf{b}B,r}\rb{\tau_{\bbf{b},r}}}.
    \label{eq:1}
\end{gather}
We have that $\sup_{\xinR}\vb{g_{\bbf{x},r}\rb{\tau_{\bbf{b},r}}}=\lambda$, which leads to two possible cases
\begin{enumerate}[leftmargin=1\parindent]
    \item $\sup_{\xinR}\vb{g_{\bbf{x},r}\rb{\tau_{\bbf{b},r}}}=\sup_{\xinR}{g_{\bbf{x},r}\rb{\tau_{\bbf{b},r}}}=\lambda$
    
    Here we use that \[\sup_{\xinR}{g_{\bbf{x},r}\rb{\tau_{\bbf{b},r}}}-\inf_{\xinR}{g_{\bbf{x},r}\rb{\tau_{\bbf{b},r}}}=\sup_{\xinR}{f_{\bbf{x}}\rb{\tau_{\bbf{b},r}}}-\inf_{\xinR}{f_{\bbf{x}}\rb{\tau_{\bbf{b},r}}}<\df ,\]
    which then implies $\inf_{\xinR}{g_{\bbf{x},r}\rb{\tau_{\bbf{b},r}}}>\lambda-\df >\lambda-h/2>0$. Then 
    \[\mathrm{sign} \sqb{g_{\bbf{x},r}\rb{\tau_{\bbf{b},r}}}=s_{\bbf{b},r}=\mathrm{sign} \sqb{g_{\bbf{b}B,r}\rb{\tau_{\bbf{b},r}}}=1.\]
    Then, continuing the derivation in \eqref{eq:1}, we get
    \begin{equation}
    g_{r+1}\rb{\tau_{\bbf{b},r}}=\sup_{\xinR}\vb{g_{\bbf{x},r}\rb{\tau_{\bbf{b},r}}-h s_{\bbf{b},r}}=\sup_{\xinR}\vb{g_{\bbf{x},r}\rb{\tau_{\bbf{b},r}}-h}.
    \label{eq:2}
    \end{equation}
    
    Next, we evaluate the sign of $g_{\bbf{x},r}\rb{\tau_{\bbf{b},r}}-h$ as follows  
    \begin{align}
        g_{\bbf{x},r}\rb{\tau_{\bbf{b},r}}-h&\leq\sup_{\xinR}g_{\bbf{x},r}\rb{\tau_{\bbf{b},r}}-h=\lambda-h,\\
        g_{\bbf{x},r}\rb{\tau_{\bbf{b},r}}-h&\geq\inf_{\xinR}g_{\bbf{x},r}\rb{\tau_{\bbf{b},r}}-h=\lambda-\df -h\\
        &\geq \lambda-\frac{3h}{2}>0.
    \end{align}    
    
    The final inequality follows from the assumption $h<2\lambda/3$. It follows that \eqref{eq:2}
    \begin{equation}
        g_{r+1}\rb{\tau_{\bbf{b},r}}=\sup_{\xinR}{g_{\bbf{x},r}\rb{\tau_{\bbf{b},r}}-h}\leq\lambda-h.
        \label{eq:3}
    \end{equation}
    
    \item $\sup_{\xinR}\vb{g_{\bbf{x},r}\rb{\tau_{\bbf{b},r}}}=-\inf_{\xinR}{g_{\bbf{x},r}\rb{\tau_{\bbf{b},r}}}=\lambda$. 
As before, here we prove that $\mathrm{sign} \sqb{g_{\bbf{x},r}\rb{\tau_{\bbf{b},r}}}=-1$, $ g_{r+1}\rb{\tau_{\bbf{b},r}}=\sup_{\xinR}\vb{g_{\bbf{x},r}\rb{\tau_{\bbf{b},r}}+h}$, and finally $g_{\bbf{x},r}\rb{\tau_{\bbf{b},r}}+h\in[-\lambda+h,0)$ which leads to \eqref{eq:3}.
\end{enumerate}
As in the \emph{base case} $r=1$, given that $\tau_{\bbf{b},r+1}=\inf \cb{x_1>\tau_{\bbf{b},r}\setsep g_{r+1}\rb{x_1}=\lambda}$ is true by definition and using the continuity of $g_{r+1}$, we get $g_{r+1}\rb{\tau_{\bbf{b},r+1}}=\lambda$. Using \eqref{eq:3} and the same reasoning as in the \emph{base case} leading to \eqref{eq:separation_MD_aux} the proposition follows. 

%% file: Proofs/WellDef.tex
For this to be true we need to show the existence of $R_{\bbf{b}}^+$ in \eqref{eq:Rmax}. First, we derive some preliminary properties of $f_{\bbf{x}}$. Due to Proposition \ref{prop:pw_for_slices} we have that $f_{\bbf{x}}\in\mathsf{PW}_{\Omega_1}\rb{\R}$ and implicitly $f_{\bbf{x}}\in L^2\rb{\R}$. 
Using the properties of the $\pwd$ space and $\norm{f}_\infty<\infty$, it follows that 
\begin{equation}
\lim_{x_1\rightarrow\infty}f_{\bbf{x}}\rb{x_1}=0,\quad \forall \bbf{x}\in\R^{D-1}.
\label{eq:lim_zero}
\end{equation}
We require to show that the limit in \eqref{eq:lim_zero} is uniform for all $\bbf{x}$, which is done in the following lemma.
\begin{lem}[Uniform convergence]
\label{lem:unif_convergence}
The following holds true
\begin{equation*}
    \forall \epsilon>0,\exists x_1^*>0,\forall x_1>x_1^*,\forall \bbf{x}\in\cR_{\bbf{b}},\quad\mathrm{s.t.}\quad \vb{f_{\bbf{x}}\rb{x_1}}<\epsilon.
\end{equation*}
\end{lem}
\begin{proof}
We prove by contradiction. Thus we assume
\begin{equation*}
    \exists \epsilon>0,\forall x_1^*>0,\exists x_1>x_1^*,\exists \bbf{x}\in\cR_{\bbf{b}},\quad\mathrm{s.t.}\quad \vb{f_{\bbf{x}}\rb{x_1}}\geq\epsilon.
\end{equation*}
We select $x_1^*=n, x_1=x_{1,n}, \bbf{x}=\bbf{x}_n\in\cR_{\bbf{b}}$ which leads to
\begin{equation}
    \exists \epsilon>0,\forall n\in\mathbb{N}^*,\exists x_{1,n}>n,\exists \bbf{x}_n\in\cR_{\bbf{b}},\quad\mathrm{s.t.}\quad \vb{f_{\bbf{x}_n}\rb{x_{1,n}}}\geq\epsilon.
    \label{eq:sequence}
\end{equation}
We use that, because $\cR_{\bbf{b}}$ is bounded, then any sequence $\bbf{x}_n\in\cR_{\bbf{b}}$ has a subsequence $\bbf{x}_{\kappa_m}\in\cR_{\bbf{b}}$ that converges to a point in the closure of $\cR_{\bbf{b}}$, i.e., 
\begin{equation*}
    \lim_{m\rightarrow\infty} \bbf{x}_{\kappa_m} = \bbf{x}_{\infty}\in\mathrm{cl}\rb{\cR_{\bbf{b}}}.
\end{equation*}
We then select $n=\kappa_m$ in \eqref{eq:sequence} and get
\begin{equation*}
    \exists \epsilon>0,\forall m\in\mathbb{N}^*,\exists x_{1,\kappa_m}>\kappa_m,\exists \bbf{x}_{\kappa_m}\in\cR_{\bbf{b}},\quad\mathrm{s.t.}\quad \vb{f_{\bbf{x}_{\kappa_m}}\rb{x_{1,{\kappa_m}}}}\geq\epsilon.
\end{equation*}
Furthermore, $f$ is a Lipschitz-continuous function with constant $\norm{\bs{\Omega}}_2\cdot\norm{f}_\infty$ as shown in \eqref{eq:Lipschitz_bound}, which implies that
\begin{align*}
    \vb{f_{\bbf{x}_{\kappa_m}}\rb{x_{1,\kappa_m}}}-\vb{f_{\bbf{x}_{\infty}}\rb{x_{1,\kappa_m}}}&\leq
    \vb{f_{\bbf{x}_{\kappa_m}}\rb{x_{1,\kappa_m}}-f_{\bbf{x}_{\infty}}\rb{x_{1,\kappa_m}}}\\
    &\leq\norm{\bs{\Omega}}_2\cdot\norm{f}_\infty\cdot\norm{\bbf{x}_{\kappa_m}-\bbf{x}_\infty}_2.
\end{align*}
This allows defining the following lower bound on $\vb{f_{\bbf{x}_{\infty}}\rb{x_{1,\kappa_m}}}$
\begin{equation*}
    \vb{f_{\bbf{x}_{\infty}}\rb{x_{1,\kappa_m}}}\geq\vb{f_{\bbf{x}_{\kappa_m}}\rb{x_{1,\kappa_m}}}-\norm{\bs{\Omega}}_2\cdot\norm{f}_\infty\cdot\norm{\bbf{x}_{\kappa_m}-\bbf{x}_\infty}_2.
\end{equation*}
We know that $\vb{f_{\bbf{x}_{\kappa_m}}\rb{x_{1,\kappa_m}}}\geq\epsilon$ and $\lim_{m\rightarrow\infty}\norm{\bbf{x}_{\kappa_m}-\bbf{x}_\infty}_2 = 0$. Then it follows that $\vb{f_{\bbf{x}_{\infty}}\rb{x_{1,\kappa_m}}}$ cannot converge to $0$ for $m\rightarrow\infty$, which directly contradicts \eqref{eq:lim_zero}. Then the starting assumption is wrong, and the Lemma follows.
\end{proof}

We approach this proof by contradiction. If we assume that $R_{\bbf{b}}^+\geq0$ does not exist it follows that $\tau_{\bbf{b},r}$ in \eqref{eq:tau_recurrent} is well-defined for $r\in\Z_+$. Given that, by definition, $\tau_{\bbf{b},r}$ is increasing as a function of $r$, then due to Proposition \ref{prop:fold_sep_multid} it follows that $\lim_{r\rightarrow\infty}\tau_{\bbf{b},r}=\infty$.
We use Lemma \ref{lem:unif_convergence} for $\epsilon=\frac{h}{4}$, which yields
$\exists x_1^*>0$ such that $\vb{f_{\bbf{x}}\rb{x_1}}<\frac{h}{4}, \forall \bbf{x}\in\cR_{\bbf{b}}, \forall x_1>x_1^*$. Given our assumption on $\tau_{\bbf{b},r}$ then $\exists r^*\in\Z_+$ such that $\tau_{\bbf{b},r^*}>x_1^*$. Thus, by definition, $\sup_{\bbf{x}\in\cR_{\bbf{b}}} \vb{f_{\bbf{x}}\rb{\tau_{\bbf{b},r^*}}-\varepsilon_{\bbf{b},r^*}\rb{\tau_{\bbf{b},r^*}}}=\lambda$. 
Using \eqref{eq:1}--\eqref{eq:3} where $r$ is replaced by $r^*$ and functions $g_r, g_{\bbf{x},r}$ are defined in \eqref{eq:gxr},\eqref{eq:gr}, one can show that $g_{r^*+1}\rb{\tau_{\bbf{b},r}}\leq\lambda-h$. The next folding time $\tau_{\bbf{b},r^*+1}$ satisfies $g_{r^*+1}\rb{\tau_{\bbf{b},r^*+1}}=\lambda$. In the following we will show this is not possible. Specifically, for $x_1>\tau_{\bbf{b},r^*}$,
\begin{align*}
    g_{r^*+1}\rb{x_1}&= \sup \vb{f_{\bbf{x}}\rb{x_1}-\varepsilon_{\bbf{b},r^*}\rb{x_1}}=\sup \vb{f_{\bbf{x}}\rb{x_1}-\varepsilon_{\bbf{b},r^*}\rb{\tau_{\bbf{b},r^*}}}\\
    &=\sup \vb{f_{\bbf{x}}\rb{\tau_{\bbf{b},r^*}}-\varepsilon_{\bbf{b},r^*}\rb{\tau_{\bbf{b},r^*}}+f_{\bbf{x}}\rb{x_1}-f_{\bbf{x}}\rb{\tau_{\bbf{b},r^*}}}\\
    &\leq \sup \vb{f_{\bbf{x}}\rb{\tau_{\bbf{b},r^*}}-\varepsilon_{\bbf{b},r^*}\rb{\tau_{\bbf{b},r^*}}}+\sup\vb{f_{\bbf{x}}\rb{x_1}-f_{\bbf{x}}\rb{\tau_{\bbf{b},r^*}}}\\
    &\leq \lambda-h + h/2 = \lambda-h/2 <\lambda.
\end{align*}
We conclude that the definition of $\tau_{\bbf{b},r^*+1}$ via \eqref{eq:tau_recurrent} is therefore not possible, and thus our assumption that $\tau_{\bbf{b},r}$ is well-defined for $r\in \Z$ is false, and the proposition follows.

%% file: Proofs/ModDynRange.tex
We first assume that $x_1$ satisfies $x_1\in\left[\tau_{\bbf{b},r},\tau_{\bbf{b},r+1}\right)$ and then extrapolate to the whole real axis. Given the definition of function $\varepsilon_{\bbf{b},r} $ \eqref{eq:eps_recurrent} it follows that $\varepsilon\rb{\mathbf{Vx}}=\varepsilon_{\bbf{b},r}\rb{x_1}$ and thus
\begin{equation*}
    z\rb{\mathbf{Vx}}=f_{\bbf{x}}\rb{x_1}-\varepsilon_{\bbf{b},r}\rb{x_1}=g_{\bbf{x},r+1}\rb{x_1}.
\end{equation*}
It was shown before that $g_{\bbf{x},r+1}\rb{\tau_{\bbf{b},r}}\leq \lambda-h$ \eqref{eq:3}. Furthermore, due to the definition of $\tau_{\bbf{b},r+1}$ \eqref{eq:tau_recurrent} we get that $g_{r+1}\rb{\tau_{\bbf{b},r+1}}=\lambda$ and $\vb{g_{\bbf{x},r+1}\rb{x_1}}<\lambda$ when our assumption $x_1\in\left[\tau_{\bbf{b},r},\tau_{\bbf{b},r+1}\right)$ holds. By repeating the process above for $r\in\Z_+$ we get $\vb{z\rb{\mathbf{Vx}}}<\lambda, \forall x_1\geq0$. For $x_1<0$, the process above is reproduced for $x_1\in\left(\tau_{\bbf{b},r-1},\tau_{\bbf{b},r}\right], \forall r\in\Z_-$.

%% file: Proofs/BoundDelfB.tex
Let $\mathbf{x},\mathbf{\uchi}\in\cR_{\bbf{b}}$.
As shown before \eqref{eq:Lipschitz_bound}, $f_{\mathbf{V}}=f\rb{\mathbf{V}\cdot}$ is a Lipschitz-continuous function with constant $\norm{\mathbf{\Omega}}_2\cdot\norm{f}_\infty$. Then it follows that
\begin{align}
    \vb{f\rb{\mathbf{Vx}}-f\rb{\mathbf{V\uchi}}}&\leq \norm{\mathbf{\Omega}}_2\cdot\norm{f}_\infty\cdot\norm{\mathbf{x}-\mathbf{ \uchi}}_2\\
    &\leq \norm{\mathbf{\Omega}}_2\cdot\norm{f}_\infty\cdot B\sqrt{D}.
\end{align}

We recall that $\df$ satisfies \eqref{eq:DelfB}
\begin{equation}
\df\triangleq\sup_{\bbf{b}\in\ZDred,x_1\in\R} \sqb{\sup_{\bbf{x}\in \cR_{\bbf{b}}} f_{\bbf{x}}\rb{x_1}-\inf_{\bbf{x}\in \cR_{\bbf{b}}} f_{\bbf{x}}\rb{x_1}}.
\end{equation}

We fix a $\bbf{b}\in\ZDred$, and define $\dfb$ as
\begin{equation}
    \dfb\triangleq\sup_{x_1\in\R} \sqb{\sup_{\bbf{x}\in \cR_{\bbf{b}}} f_{\bbf{x}}\rb{x_1}-\inf_{\bbf{x}\in \cR_{\bbf{b}}} f_{\bbf{x}}\rb{x_1}}.
\end{equation}
We then use the property that there always exists a sequence in a set converging to the infimum or supremum. Therefore, $\exists \bbf{x}_n,\bbs{\uchi}_n, x_{1,n}$ such that 
\begin{equation}
    \lim_{n\rightarrow \infty}f_{\bbf{x}_n}\rb{x_{1,n}}-f_{\bbs{\uchi}_n}\rb{x_{1,n}} =\dfb.
\end{equation}
Given that $f_{\bbf{x}_n}\rb{x_{1,n}}$ converges to the supremum and $f_{\bbs{\uchi}_n}\rb{x_{1,n}}$ to the infimum, it follows that $\exists N\in\Z_+$ such that $f_{\bbf{x}_n}\rb{x_{1,n}}-f_{\bbs{\uchi}_n}\rb{x_{1,n}}\geq 0, \forall n \in \Z, n\geq N$. Then
\begin{align}
    f_{\bbf{x}_n}\rb{x_{1,n}}-f_{\bbs{\uchi}_n}\rb{x_{1,n}}&=\vb{f_{\bbf{x}_n}\rb{x_{1,n}}-f_{\bbs{\uchi}_n}\rb{x_{1,n}}},\quad \forall n\geq N\\
    &\leq \norm{f}_\infty \cdot  B\sqrt{D} \cdot \sqrt{\sum_{d=1}^D \Omega_d^2}.
\end{align}
By taking $n\rightarrow\infty$ above we have that $\dfb\leq \norm{f}_\infty \cdot  B\sqrt{D} \cdot \sqrt{\sum_{d=1}^D \Omega_d^2},\quad\forall\bbf{b}\in\ZDred$, and computing $\df=\sup_{\bbf{b}\in\ZDred} \dfb$ leads to the desired bound.

%% file: Proofs/FoldDet.tex
As in the continuous-time scenario, we denote the \D slices of the samples in each band $\bbf{b}$ by
\begin{align}
    \gamma_{\bbf{k}}\sqb{k_1}&=\gamma\sqb{k_1T_1\mathbf{v}_1+\sum_{d=2}^D k_d T_d \mathbf{v}_d},\\
    \varepsilon_{\gamma,\bbf{k}}\sqb{k_1}&=\varepsilon_\gamma\sqb{k_1T_1\mathbf{v}_1+\sum_{d=2}^D k_d T_d \mathbf{v}_d},\\
    \eta_{\bbf{k}}\sqb{k_1}&=\eta\sqb{k_1T_1\mathbf{v}_1+\sum_{d=2}^D k_d T_d \mathbf{v}_d}.
\end{align}
where $\bbf{T}\bbf{k}\in\cR_{\bbf{b}}$ and $\bbf{k}=\sqb{k_2,\dots,k_D}\in\ZDred$. We note that, due to Definition \ref{def:cont_moduloMD}, $ \varepsilon_{\gamma,\bbf{k}}\sqb{k_1}$ does not change with $\bbf{k}$ as long as $\bbf{T}\bbf{k}\in\cR_{\bbf{b}}$. The filtered samples $y_{\bbf{b},m}$ satisfy
\begin{equation*}
    y_{\bbf{b},m}=\inner{y,\psi_{\bbf{b},m}} =\inner{\gamma,\psi_{\bbf{b},m}}-\inner{\varepsilon_{\gamma},\psi_{\bbf{b},m}}+\inner{\eta,\psi_{\bbf{b},m}}
\end{equation*}    
We first exploit that $\gamma_{\bbf{k}}\sqb{k_1}=f_{\bbf{T}\bbf{k}}\rb{k_1T_1}$ where $f_{\bbf{x}}$ is bandlimited to $\Omega_1\ \mathrm{rad/s}$, which yields \cite{Bhandari:2020:Ja,Florescu:2022:J}
\begin{align}
\begin{split}
    \vb{\inner{\gamma,\psi_{\bbf{b},m}}}&=
    \bigg\lvert\tfrac{1}{N^B}\sum_{\bbf{T}\bbf{k}\in\cR_{\bbf{b}}}\inner{\gamma_{\bbf{k}},\Delta^N\sqb{\cdot-m}}\bigg\rvert\\
    &\leq\tfrac{1}{N^B}\sum_{\bbf{T}\bbf{k}\in\cR_{\bbf{b}}}\vb{\inner{\gamma_{\bbf{k}},\Delta^N\sqb{\cdot-m}}}\\
    &\leq\tfrac{1}{N^B}\sum_{\bbf{T}\bbf{k}\in\cR_{\bbf{b}}} \rb{T_1\Omega_1 e}^N \norm{f}_\infty=\rb{T_1\Omega_1 e}^N \norm{f}_\infty.
\end{split}
\label{eq:gamma_bound}
\end{align}
Similarly, for the residual samples $\varepsilon_\gamma$ the following holds
\begin{align}
\begin{split}
    \vb{\inner{\varepsilon_\gamma,\psi_{\bbf{b},m}}}&=
    \bigg\lvert\tfrac{1}{N^B}\sum_{\bbf{T}\bbf{k}\in\cR_{\bbf{b}}}\inner{\varepsilon_{\gamma,\bbf{k}},\Delta^N\sqb{\cdot-m}}\bigg\rvert\\
    &\leq\tfrac{1}{N^B}\sum_{\bbf{T}\bbf{k}\in\cR_{\bbf{b}}}\vb{\inner{\varepsilon_{\gamma,\bbf{k}},\Delta^N\sqb{\cdot-m}}}=\vb{\inner{\varepsilon_{\gamma,\bbf{k}^*},\Delta^N\sqb{\cdot-m}}},
\end{split}
\label{eq:eps_bound1}
\end{align}    
for all $ \bbf{k}^*$ satisfying $\bbf{T}\bbf{k}^*\in\cR_{\bbf{b}}$, given that $\varepsilon_{\gamma,\bbf{k}}$ does not change within the band $\bbf{b}$ as a function of $\bbf{k}$ \eqref{eq:eps_recurrent}. We note that $\mathrm{supp} \rb{\Delta^N\sqb{\cdot-m}}=\cb{m,\dots,m+N}$ and derive that 
\begin{equation}
\left\lbrace
    \begin{array}{cc}
        \vb{\inner{\varepsilon_\gamma,\psi_{\bbf{b},m}}}\geq h & \mathrm{if}\ m\in\mathbb{S}_N, \\
        \vb{\inner{\varepsilon_\gamma,\psi_{\bbf{b},m}}}=0 & \mathrm{otherwise}, 
    \end{array}\right.
    \label{eq:eps_bound2}
\end{equation}
Therefore, if $\rb{T_1\Omega_1 e}^N \norm{f}_\infty<h/2$ and $\vb{\inner{\eta,\psi_{\bbf{b},m}}}<h/2-\rb{T_1\Omega_1 e}^N \norm{f}_\infty$, then, $\mathrm{if}\ m\in\mathbb{S}_N$,
\begin{align}
\begin{split}
    \vb{\inner{y,\psi_{\bbf{b},m}}} &=\vb{\inner{\varepsilon_{\gamma},\psi_{\bbf{b},m}}-\rb{\inner{\gamma,\psi_{\bbf{b},m}}+\inner{\eta,\psi_{\bbf{b},m}}}}\\
    &\geq \vb{\inner{\varepsilon_{\gamma},\psi_{\bbf{b},m}}}-\vb{\inner{\gamma,\psi_{\bbf{b},m}}+\inner{\eta,\psi_{\bbf{b},m}}}\\
    &\geq h-h/2=h/2.
\end{split}
\label{eq:y_bound1}
\end{align}
Furthermore, for $m\not\in\mathbb{S_N}$, we get $\vb{\inner{y,\psi_{\bbf{b},m}}} \leq\vb{\inner{\gamma,\psi_{\bbf{b},m}}}+\vb{\inner{\eta,\psi_{\bbf{b},m}}}< h/2$.
Then, we identify if  $m\in\mathbb{S}_N$ by thresholding sequence $\inner{y,\psi_{\bbf{b},m}}$ via the following inequalities.
\begin{equation}
\left\lbrace
    \begin{array}{cc}
        \vb{\inner{y,\psi_{\bbf{b},m}}} \geq h/2 & \mathrm{if}\ m\not\in\mathbb{S}_N, \\
        \vb{\inner{y,\psi_{\bbf{b},m}}}< h/2   & \mathrm{otherwise}, 
    \end{array}\right.
    \label{eq:y_bound2}
\end{equation}    

Sequence $\eta\sqb{\bbf{k}}$ is drawn from the normal distribution and is not bounded, therefore we can only guarantee that $\vb{\inner{\eta,\psi_{\bbf{b},m}}}<h/2-\rb{T_1\Omega_1 e}^N \norm{f}_\infty$ holds with a given probability. However, we will show that modulo-hysteresis allows increasing this probability exponentially.

We use two properties of the \pdf of Gaussian distributions. First, given $M$ random variables $\eta_i\sim\mathcal{N}\rb{0,\sigma_i^2},$ $ i=1,\dots,M$ their summation satisfies $\sum_{i=1}^M \eta_i \sim\mathcal{N}\rb{0,\sigma^2}$, where $\sigma=\sqrt{\sigma_1^2+\dots+\sigma_{M}^2}$. Second, for a random variable $\eta_1\sim\mathcal{N}\rb{0,\sigma^2}$, the multiplication with a constant $\alpha$ yields  $\alpha\eta_1\sim\mathcal{N}\rb{0,\rb{\alpha\sigma}^2}$.

Then the noise term $\inner{\eta,\psi_{\bbf{b},m}}$ satisfies
\begin{align}
\begin{split}
    \inner{\eta,\psi_{\bbf{b},m}}&=\sum_{\bbf{k}\in\ZDred} \sum_{k_1\in\Z} \frac{\Delta^N\sqb{k_1-m}\cdot \eta\sqb{\mathbf{k}}}{N^B}\\
    &=\sum_{k_1\in\Z} \Delta^N\sqb{k_1-m}\cdot \bar{\eta}\sqb{k_1}, \quad \bar{\eta}\sqb{k_1}=\frac{1}{N^B}\sum_{\bbf{k}\in\ZDred}\eta\sqb{\mathbf{k}}.
\end{split}
\label{eq:noise1}
\end{align}
Using the two \pdf properties above, $\bar{\eta}\sim\mathcal{N}\rb{0,\rb{\sigma\sqrt{N^B}/N^B}^2}=\mathcal{N}\rb{0,\frac{\sigma^2}{N^B}}$. As expected, averaging gradually narrows down the \pdf of the distribution around the origin. Furthermore, we can write 
\begin{equation}
    \inner{\eta,\psi_{\bbf{b},m}}=\Delta^N_- \ast \bar{\eta}\sqb{m},\quad \Delta^N_-\sqb{k_1}=\Delta^N\sqb{-k_1},\forall k_1\in\Z.
    \label{eq:noise2}
\end{equation}
Given that $\Delta_-^n=\Delta_-^1\ast\dots\ast\Delta_-^{n-1}, \forall n\in\Z$, we will compute recursively the \pdf of  $\inner{\eta,\psi_{\bbf{b},m}}$ as follows. First, $\Delta_-^1\ast \bar{\eta}\sqb{m}=\bar{\eta}\sqb{m+1}-\bar{\eta}\sqb{m} $. Given that $-\bar{\eta}\sqb{m}\sim\mathcal{N}\rb{0,\frac{\sigma^2}{N^B}}$, we get $\Delta_-^1\ast \bar{\eta}\sqb{m}\sim\mathcal{N}\rb{0,\sigma^2{\frac{2}{N^B}}}$. Recursively, 
\begin{equation}
    \Delta_-^N\ast \bar{\eta}\sqb{m}\sim\mathcal{N}\rb{0,\sigma^2{\frac{2^N}{N^B}}}.
    \label{eq:noise3}
\end{equation}

In the equation above one can notice that the finite difference degree $N$ leads to an exponential increase in the standard deviation of the noise term. A very similar result was reported for bounded noise in the \D case \cite{Bhandari:2020:Ja,Florescu:2022:J,Florescu:2022:Cb}. However, in this \MD case we have the option to decrease the noise by increasing the number of samples $N^B=\prod_{d=2}^D \frac{B}{T_d}$. This can be done either by increasing the band size $B$ within the allowable range ensuring $\df <\min\cb{h/2,2\lambda-3h}$, but also by decreasing the sampling periods $T_d, d=2,\dots,D$. Both of these act only on dimensions $x_2,\dots,x_D$ and are fully independent of dimension $x_1$.

Therefore, the noise term $\inner{\eta,\psi_{\bbf{b},m}}$ in \eqref{eq:filtered_data} represents a random variable that allows to correctly evaluate if $m\in\mathbb{S}_N$ via \eqref{eq:y_bound2} when $\inner{\eta,\psi_{\bbf{b},m}}<h/2-\rb{T_1\Omega_1 e}^N \norm{f}_\infty$. The probability that this doesn't hold is denoted by $p_{\mathsf{err}}$ which is calculated using the \pdf of the normal distribution as
\begin{equation}
    p_{\mathsf{err}}=2\int_{\eta_{\mathsf{max}}}^\infty \frac{1}{\sigma_0\sqrt{2\pi}}\cdot e^{-\frac{1}{2}\rb{\frac{x}{\sigma_0}}^2}dx,
    \label{eq:perr}
\end{equation}
where $\sigma_0=\sigma\cdot\sqrt{\frac{2^N}{N^B}}$ and $\eta_{\mathsf{max}}=h/2-\rb{T_1\Omega_1 e}^N \norm{f}_\infty$. The integral above can be bounded in terms of the \emph{complementary error function} as follows, which finalizes the proof \cite{Chiani:2003:J,Jacobs:1965:B}
\begin{equation}
    p_{\mathsf{err}}\leq e^{-\rb{\frac{\eta_{\mathsf{max}}}{\sigma_0\sqrt{2}}}^2}.
    \label{eq:perr2}
\end{equation}

%% file: Proofs/MbDet.tex
We define by $y_{\bbf{b},\bbf{b}^*}$ the samples $y$ filtered with $\psi_{\bbf{b},\bbf{b}^*}$ where $\bbf{b}^* \in\nbr$, such that $\sqb{\bbf{b}^*}_{d^*}=\sqb{\bbf{b}}_{d^*}+1$ and
\begin{equation*}
    y_{\bbf{b},\bbf{b}^*}=\inner{y,\psi_{\bbf{b},\bbf{b}^*}} =\inner{\gamma,\psi_{\bbf{b},\bbf{b}^*}}-\inner{\varepsilon_{\gamma},\psi_{\bbf{b},\bbf{b}^*}}+\inner{\eta,\psi_{\bbf{b},\bbf{b}^*}}
\end{equation*}    
Along the same lines as \eqref{eq:gamma_bound}, we derive
\begin{equation}
     \vb{\inner{\gamma,\psi_{\bbf{b},\bbf{b}^*}}}\leq \rb{T_{d^*}\Omega_{d^*}e}^N\norm{f}_\infty.
\end{equation}
Along dimension $x_{d^*}$ and for $\bbf{T}\bbf{k}\in\cR_{\bbf{b}}\cup\cR_{\bbf{b}^*}, k_1=0$, the support of $\psi_{\bbf{b},\bbf{b}^*}$ is 
\[\cb{N_{d^*}^B\sqb{\bbf{b}^*}_{d^*}-1,\dots,N_{d^*}^B\sqb{\bbf{b}^*}_{d^*}-1+N}.\]
Furthermore, the discontinuity between the bands would be located in between the samples $\rb{N_{d^*}^B\sqb{\bbf{b}^*}_{d^*}-1}T_{d^*},N_{d^*}^B\sqb{\bbf{b}^*}_{d^*}T_{d^*}$. Using the same reasoning as before (\ref{eq:eps_bound1}-\ref{eq:eps_bound2})
\begin{equation}
\left\lbrace
    \begin{array}{cc}
        \vb{\inner{\varepsilon_\gamma,\psi_{\bbf{b},\bbf{b}^*}}}\geq h & \mathrm{if}\ M_{\bbf{b}}\neq M_{\bbf{b}^*}, \\
        \vb{\inner{\varepsilon_\gamma,\psi_{\bbf{b},\bbf{b}^*}}}=0 & \mathrm{otherwise}, 
    \end{array}\right.
    \label{eq:eps111}
\end{equation}

Therefore, if $\rb{T_{d^*}\Omega_{d^*} e}^N \norm{f}_\infty<h/2$ and $\vb{\inner{\eta,\psi_{\bbf{b},\bbf{b}^*}}}<h/2-\rb{T_{d^*}\Omega_{d^*} e}^N \norm{f}_\infty$, then, $\mathrm{if}\ M_{\bbf{b}}\neq M_{\bbf{b}^*}$, as before, we get \eqref{eq:y_bound1}
\begin{equation*}
    \vb{\inner{y,\psi_{\bbf{b},\bbf{b}^*}}}\geq h/2.
\end{equation*}
Therefore, as before \eqref{eq:y_bound2}, we identify if  $M_{\bbf{b}}\neq M_{\bbf{b}^*}$ by thresholding sequence $\inner{y,\psi_{\bbf{b},\bbf{b}^*}}$ via the following inequalities:
\begin{equation}
\left\lbrace
    \begin{array}{cc}
        \vb{\inner{y,\psi_{\bbf{b},\bbf{b}^*}}} \geq h/2 & \mathrm{if}\ M_{\bbf{b}}\neq M_{\bbf{b}^*}, \\
        \vb{\inner{y,\psi_{\bbf{b},\bbf{b}^*}}}< h/2   & \mathrm{otherwise}.
    \end{array}\right.
    \label{eq:threshold1}
\end{equation}   

Assuming that $M_{\bbf{b}}\neq M_{\bbf{b}^*}$, we compute $\mathrm{sign} \rb{\inner{y,\psi_{\bbf{b},\bbf{b}^*}}}$ as follows. We first show that 
\begin{equation}
\mathrm{sign} \rb{\inner{y,\psi_{\bbf{b},\bbf{b}^*}}}=-\mathrm{sign} \rb{\inner{\varepsilon_{\gamma},\psi_{\bbf{b},\bbf{b}^*}}}
\label{eq:sign1}
\end{equation}
If we assume by contradiction that $\mathrm{sign} \rb{\inner{y,\psi_{\bbf{b},\bbf{b}^*}}}=\mathrm{sign} \rb{\inner{\varepsilon_{\gamma},\psi_{\bbf{b},\bbf{b}^*}}}$ we get
$\vb{\inner{y,\psi_{\bbf{b},\bbf{b}^*}}+\inner{\varepsilon_\gamma,\psi_{\bbf{b},\bbf{b}^*}}}=\vb{\inner{\gamma,\psi_{\bbf{b},\bbf{b}^*}}}<h/2$.
However, from \eqref{eq:eps111}-\eqref{eq:threshold1} we have that $\vb{\inner{y,\psi_{\bbf{b},\bbf{b}^*}}}\geq h/2$ and $\vb{\inner{\varepsilon_\gamma,\psi_{\bbf{b},\bbf{b}^*}}}>0$. Given our assumption that the two quantities have the same sign, we get a contradiction and thus \eqref{eq:sign1} is true.

Using the expression of $\psi_{\bbf{b},\bbf{b}^*}$ in \eqref{eq:psibb}, the definition of the residual in \eqref{eq:eps} and Proposition \ref{prop:variation_Mb}, it can be shown directly that
\[\inner{\varepsilon_{\gamma},\psi_{\bbf{b},\bbf{b}^*}}=h\rb{M_{\bbf{b}^*}- M_{\bbf{b}}}\rb{-1}^{N+1}\Rightarrow\mathrm{sign} \rb{\inner{\varepsilon_{\gamma},\psi_{\bbf{b},\bbf{b}^*}}}=\rb{M_{\bbf{b}^*}-M_{\bbf{b}}} \rb{-1}^{N+1},\] 
and thus $\mathrm{sign} \rb{\inner{y,\psi_{\bbf{b},\bbf{b}^*}}}=\rb{M_{\bbf{b}^*}-M_{\bbf{b}}} \rb{-1}^{N}$. Finally, we get that
\begin{equation}
\begin{cases}
        M_{\bbf{b}^*}=M_{\bbf{b}}+\mathrm{sign} \rb{\inner{y,\psi_{\bbf{b},\bbf{b}^*}}}\rb{-1}^N & \mathrm{if}\ \vb{\inner{y,\psi_{\bbf{b},\bbf{b}^*}}} \geq h/2, \\
        M_{\bbf{b}^*}=M_{\bbf{b}}   & \mathrm{otherwise}.
\end{cases}
\end{equation}   

Furthermore, the filtered noise satisfies (\ref{eq:noise1}-\ref{eq:noise3})
\begin{equation}
    \inner{\eta,\psi_{\bbf{b},\bbf{b}^*}}\sim \mathcal{N}\rb{0,\sigma^2\frac{2^N}{N^{B*}}}.
\end{equation}
Then, the probability that $\vb{\inner{\eta,\psi_{\bbf{b},\bbf{b}^*}}}<h/2-\rb{T_{d^*}\Omega_{d^*} e}^N \norm{f}_\infty$ doesn't hold, denoted by $p_{\mathsf{err}}$, satisfies
\begin{equation}
    p_{\mathsf{err}}\leq e^{-\rb{\frac{\eta_{\mathsf{max}}}{\sigma_0\sqrt{2}}}^2}.
\end{equation}
where $\sigma_0=\sigma\cdot\sqrt{\frac{2^N}{N^{B*}}}$ and $\eta_{\mathsf{max}}=h/2-\rb{T_{d^*}\Omega_{d^*} e}^N \norm{f}_\infty$.

%% file: Proofs/InputRecClean.tex
We note that, for $\sigma\rightarrow0$, the results in Proposition \ref{prop:fold_det} and \ref{prop:Mb_det} hold true with probability $1$. We then compute $\widetilde{M}_{\bbf{b}}$ using Proposition \ref{prop:Mb_det} as follows. Given that $\widetilde{M}_{\bbf{0}}=0$ by definition, one can compute successively $\widetilde{M}_{\sqb{b_2,0,\dots,0}}$ from $\widetilde{M}_{\sqb{b_2-1,0,\dots,0}}$ for $\forall b_2\in\Z_+$ and subsequently $\widetilde{M}_{\sqb{b_2,0,\dots,0}}$ from $\widetilde{M}_{\sqb{b_2+1,0,\dots,0}}$ for $\forall b_2\in\Z_-$. Repeating the process for $b_d,d\in\cb{3,\dots,D}$ yields $\widetilde{M}_{\bbf{b}}, \forall \bbf{b}\in\ZDred$.

To compute the folding times via Proposition \ref{prop:fold_det}, we require that the sets characterized by each folding time in $\mathbb{S}_N$ do not overlap. Specifically we require that
\[\cb{\ceil{\frac{\tau_{\bbf{b},r_1}}{T_1}}-N,\dots,\ceil{\frac{\tau_{\bbf{b},r_1}}{T_1}}}\cap \cb{\ceil{\frac{\tau_{\bbf{b},r_2}}{T_1}}-N,\dots,\ceil{\frac{\tau_{\bbf{b},r_2}}{T_1}}}=\emptyset,\]
for $\forall r_1,r_2\in\Z, r_1\neq r_2$.
A sufficient condition for this is
\begin{equation}
\label{eq:ceil_lower_bound}
\ceil{\frac{\tau_{\bbf{b},r_2}}{T_1}}-N-\ceil{\frac{\tau_{\bbf{b},r_1}}{T_1}}>{\frac{\tau_{\bbf{b},r_2}}{T_1}}-{\frac{\tau_{,\bbf{b},r_1}}{T_1}}-\rb{N+1}>0,
\end{equation}
which can be guaranteed via Proposition \ref{prop:fold_sep_multid} if 
\begin{equation}
    \rb{N+1}T_1<\frac{h}{\Omega_1 \norm{f}_\infty}.
    \label{eq:separation_suff_cond}
\end{equation}
Without reducing the generality we first assume that $x_1\geq0$ and thus $r,k_1\geq0$. As before, the case $x_1\leq0$ is treated as a mirrored version of $x_1\geq0$.
An immediate consequence of \eqref{eq:ceil_lower_bound} is that for $r_1=0,r_2=1\Rightarrow \ceil{\frac{\tau_{\bbf{b},1}}{T_1}}\geq N+1$. Because there is no actual jump taking place at $\tau_{\bbf{b},0}=0 \Rightarrow \vb{\inner{\varepsilon_\gamma,\psi_{\bbf{b},0}}}=0$ and $\vb{\inner{y,\psi_{\bbf{b},0}}}<h/2$ via Proposition \ref{prop:fold_det}. The smallest $m$ for which filtered output satisfies $\vb{\inner{y,\psi_{\bbf{b},m}}}\geq h/2$ is $m=m_{\mathsf{min}}^1\triangleq \ceil{\frac{\tau_{\bbf{b},1}}{T_1}}-N$. The last index $m$ corresponding to folding time $\tau_{\bbf{b},1}$ detected via $\vb{\inner{y,\psi_{\bbf{b},m}}}\geq h/2$ is $m=m_{\mathsf{max}}^1\triangleq\ceil{\frac{\tau_{\bbf{b},1}}{T_1}}$. We can compute $m_{\mathsf{min}}^1$ and $m_{\mathsf{max}}^1$ as
\begin{align}
    \label{eq:m_min1}
    m_{\mathsf{min}}^1&=\min\cb{m>0\setsep\vb{\inner{y,\psi_{\bbf{b},m}}}\geq h/2},\\
    m_{\mathsf{max}}^1&=m_{\mathsf{min}}^1+N.
\end{align}
Assuming \eqref{eq:separation_suff_cond} to be true, one can then compute recursively sequences $m_{\mathsf{min}}^r, m_{\mathsf{max}}^r$ corresponding to folding time $\tau_{\bbf{b},r}$ as follows
\begin{align}
    m_{\mathsf{min}}^r&=\min\cb{m>m_{\mathsf{min}}^{r-1}+N\setsep\vb{\inner{y,\psi_{\bbf{b},m}}}\geq h/2},\\
    m_{\mathsf{max}}^r&=m_{\mathsf{min}}^r+N.
\end{align}
The folding time is estimated as $\widetilde{\tau}_{\bbf{b},r}=\sqb{m_{\mathsf{min}}^r+N}\cdot T_1$. As in the case $r=1$ we can show that $m_{\mathsf{min}}^r=\ceil{\frac{\tau_{\bbf{b},r}}{T_1}}-N$. Therefore $\widetilde{\tau}_{\bbf{b},r}=\ceil{\frac{\tau_{\bbf{b},r}}{T_1}}\cdot T_1$. Even though the folding time is not perfectly computed, this has no effect on the input recovery because $\ceil{\frac{\widetilde{\tau}_{\bbf{b},r}}{T_1}}=\ceil{\frac{{\tau}_{\bbf{b},r}}{T_1}}$ and we only evaluate the residual at the sampling locations $\mathbf{kT}$. This means that replacing ${\tau}_{\bbf{b},r}$ by $\ceil{\frac{{\tau}_{\bbf{b},r}}{T_1}}\cdot T_1$ in the expression of $\varepsilon_f\rb{\mathbf{VTk}}$  yields the same values (see Definition \ref{def:cont_moduloMD})
\begin{align}
\begin{split}
\varepsilon_\gamma\sqb{\mathbf{k}}&=h\sqb{M_{\bbf{b}}+\sum_{i=0}^r {s}_{\bbf{b},r} \ind_{\left[{\tau}_{\bbf{b},r},\infty\right)}\rb{k_1 T_1}}\\
&=h\sqb{M_{\bbf{b}}+\sum_{i=0}^r {s}_{\bbf{b},r} \ind_{\left[\ceil{\frac{{\tau}_{\bbf{b},r}}{T_1}}\cdot T_1,\infty\right)}\rb{k_1 T_1}}.
\end{split}
\label{eq:discrete_residual}
\end{align}
We note that, as explained before, we do not recover $M_{\bbf{b}}$ but $\widetilde{M}_{\bbf{b}}=M_{\bbf{b}}-M_{\bbf{0}}$. This will be accounted for at the final input reconstruction stage.

Furthermore, we estimate the sign as $\widetilde{s}_{\bbf{b},r}=-\mathrm{sign}\inner{y,\psi_{\bbf{b},m_{\mathsf{min}}^r}}$. We will show that $\widetilde{s}_{\bbf{b},r}={s}_{\bbf{b},r}$ as follows. Given that $\vb{\inner{\gamma,\psi_{\bbf{b},m_{\mathsf{min}}^r}}}<h/2$ and $\vb{\inner{y,\psi_{\bbf{b},m_{\mathsf{min}}^r}}}\geq h/2$, then, via \eqref{eq:filtered_data}, it follows that 
\[\mathrm{sign}\inner{y,\psi_{\bbf{b},m_{\mathsf{min}}^r}}=-\mathrm{sign}\inner{\varepsilon_\gamma,\psi_{\bbf{b},m_{\mathsf{min}}^r}}.\]
We use the fact that $\varepsilon_\gamma\sqb{\mathbf{k}}$ does not change for $\bbf{k}\bbf{T}\in\cR_{\bbf{b}}$. For $N\geq1$, using the expression of $\Delta^N$ and \eqref{eq:discrete_residual},
\begin{align}
\begin{split}
    \inner{\varepsilon_\gamma,\psi_{\bbf{b},m_{\mathsf{min}}^r}}&=\inner{\varepsilon_\gamma,\Delta^N\sqb{\cdot-m_{\mathsf{min}}^r}}\\
    &=h {s}_{\bbf{b},r} \sum_{k_1\in\Z}  \ind_{\left[\ceil{\frac{{\tau}_{\bbf{b},r}}{T_1}}\cdot T_1,\infty\right)}\rb{k_1 T_1}\cdot\Delta^N\sqb{k_1-m_{\mathsf{min}}^r}\\
    &=h {s}_{\bbf{b},r} \sum_{k_1\in\Z}  \ind_{\left[\rb{m_{\mathsf{min}}^r+N}\cdot T_1,\infty\right)}\rb{k_1 T_1}\cdot\Delta^N\sqb{k_1-m_{\mathsf{min}}^r}
\end{split}
\end{align}
By applying the change of variable $k_1^*=k_1-m_{\mathsf{min}}^r-N$ 
\begin{align}
\inner{\varepsilon_\gamma,\psi_{\bbf{b},m_{\mathsf{min}}^r}}
&=h {s}_{\bbf{b},r} \sum_{k_1^*\in\Z}  \ind_{\left[0,\infty\right)}\rb{k_1^* T_1}\cdot\Delta^N\sqb{k_1^*+N}\\
&=h {s}_{\bbf{b},r} \sum_{k_1^*\in\Z_+}  \Delta^N\sqb{k_1^*+N}=h {s}_{\bbf{b},r}.
\end{align}
The last equality can be shown recursively via direct calculation for $N\geq1$, given that $k_1\geq0$, which proves that ${s}_{\bbf{b},r}=\widetilde{s}_{\bbf{b},r}$.

After the folding times and signs are computed as above for all $\bbf{b}\in\ZDred$, the input samples are reconstructed as
\begin{equation}
\label{eq:input_rec}
\widetilde{\gamma}\sqb{\mathbf{k}}=y\sqb{\mathbf{k}}+\widetilde{\varepsilon}_\gamma\sqb{\mathbf{k}},
\end{equation}
where $\widetilde{\varepsilon}_\gamma\sqb{\mathbf{k}}=h\sqb{\widetilde{M}_{\bbf{b}}+\sum_{i=0}^r {s}_{\bbf{b},r} \ind_{\left[{\tau}_{\bbf{b},r},\infty\right)}\rb{k_1 T_1}}=\varepsilon_{\gamma}\sqb{\mathbf{k}}-h{M}_{\bbf{0}}$, which leads to $\widetilde{\gamma}\sqb{\mathbf{k}}=\gamma\sqb{\mathbf{k}}-h{M}_{\bbf{0}}$.